\documentclass[10pt]{article}

\usepackage{color}
\usepackage{amsmath,enumerate}
\usepackage{amssymb}
\usepackage{epsfig}
\usepackage{subfig}
\usepackage{pgf}
\usepackage{pgfcore}
\usepackage{pgfbaseimage}
\usepackage{pgfbaselayers}
\usepackage{pgfbasepatterns}
\usepackage{pgfbaseplot}
\usepackage{pgfbaseshapes}
\usepackage{pgfbasesnakes}
\usepackage{tikz}
\usepackage{rotating}
\usepackage{multicol}
\usepackage[hmargin=2cm,vmargin=2cm]{geometry}

\def\C{\mathcal{C}}

\newenvironment{proof}{\par \noindent \textsc{Proof.} }{\hfill$\Box$\medskip}

\newcommand{\Z}{\mathbb{Z}}
\renewcommand{\S}{\mathcal{S}}
\newcommand{\D}{D}
\newcommand{\Lpat}{\tikz[scale = 0.3] \draw[mygrid]  (1,0) node[gridnode] {} -- (0,0) node[gridnode] {} -- (0,1)  node[gridnode] {}; }

\newtheorem{theorem}{Theorem}
\newtheorem{proposition}[theorem]{Proposition}
\newtheorem{lemma}[theorem]{Lemma}

\definecolor{roug}{RGB}{0,0,0} 
\definecolor{orang}{RGB}{255,153,0}
\definecolor{fond}{RGB}{255,255,230}
\definecolor{coul1}{RGB}{255,255,50}
\definecolor{coul2}{RGB}{255,100,0}
\definecolor{coul3}{RGB}{9,156,9}
\definecolor{coul4}{RGB}{0,0,255}
\definecolor{coul5}{RGB}{0,255,255}
\definecolor{coul6}{RGB}{255,0,255}

\tikzstyle{gridnode}=[shape=circle,fill=black,draw=black,minimum size=0.5pt,inner sep=0.5pt]
\tikzstyle{neighbour}=[shape=circle,fill=black,draw=black,minimum size=0.5pt,inner sep=1.3pt]
\tikzstyle{centernode}=[shape=rectangle, fill=white, line width=2pt, draw=roug,minimum size=0.5pt,inner sep=2pt]
\tikzstyle{code}=[shape=circle,fill=black,draw=black,minimum size=0.5pt,inner sep=2.2pt]
\tikzstyle mygrid=[line width=1,color=black!20]
\tikzstyle ball=[thick,roug]
\tikzstyle{pattern}=[draw=roug,thick,rectangle,inner sep=7pt,fill=roug!20, semitransparent]
\tikzstyle{patternpath}=[draw=roug,thick,fill=roug!20, semitransparent]
\tikzstyle{centernode2}=[shape=circle,fill=roug,draw=roug,minimum size=0.5pt,inner sep=1.3pt]
\tikzstyle{labelnode}=[circle,fill=white, inner sep =0]
\tikzstyle{labelnodep}=[circle,fill=roug!20, inner sep =0]
\tikzstyle{evennode}=[shape=circle,fill=white, draw=black,minimum size=0.5pt,inner sep=2.5pt]
\tikzstyle{oddnode}=[shape=circle,fill=black!50,draw=black,minimum size=0.5pt,inner sep=2.5pt]

\title{Tolerant identification with Euclidean balls\thanks{This research is
    supported by the ANR Project IDEA {\scriptsize $\bullet$} {ANR-08-EMER-007},  2009-2011.}
}
\author{Ville Junnila\thanks{Corresponding author.
Department of Mathematics, University of Turku,
20014 Turku, Finland, e-mail: ville.junnila@utu.fi, phone: +358 2 333 6021, fax: +358 2 231 0311 }
, Tero Laihonen\thanks{
Department of Mathematics, University of Turku,
20014 Turku, Finland, e-mail: tero.laihonen@utu.fi}  \ and Aline Parreau\thanks{
Institut Fourier, 100 rue des maths, BP 74, 38402 St Martin d'H\`eres cedex, France, e-mail: aline.parreau@ujf-grenoble.fr}}
\date{\today}

\begin{document}

\maketitle

\renewcommand{\abstractname}{Abstract}

\begin{abstract}
The concept of identifying codes was introduced by Karpovsky, Chakrabarty and Levitin in 1998. The identifying codes can be applied, for example, to sensor networks. In this paper, we consider as sensors the set $\Z^2$ where one sensor can check its neighbours within Euclidean distance $r$. We construct tolerant identifying codes in this network that are robust against some changes in the neighbourhood monitored by each sensor. We give bounds for the smallest density of a tolerant identifying code for general values of $r$. We also provide infinite families of values $r$ with optimal such codes and study the case of small values of $r$.
\end{abstract}

\section{Introduction}
Let a network be modelled by a simple, connected and undirected graph $G=(V,E)$ with vertex set $V$ and edge set $E$. We can place a sensor in any vertex $u$. A sensor is able to check its closed neighbourhood $N[u]$ (i.e., the adjacent vertices and itself) and report to a central controller if it detects something wrong there (like a smoke detector). The idea is to place as few sensors as possible in such a way that we could uniquely determine where (that is, in which vertex) the problem occurs (if any) knowing only the set of sensors which gave us the alarm.

Let us denote the subset of vertices, where we placed the sensors, by $C$. In order to find the sought object (like fire) in our network, we need to choose $C$ in the following way. Denote the set of sensors monitoring a vertex $u\in V$ by $ I(u)=N[u]\cap C$. Suppose that $C$ satisfies the following two conditions: $I(u)\neq \emptyset$ for every $u\in V$ and $I(u)\neq I(v)$ for all $u,v\in V$, $u\neq v$. Hence, $I(u)$ is the set of sensors giving the alarm if there is a problem in $u$, and since it is unique and nonempty for each $u\in V$, we can determine the vertex with a problem (if there is any). Such a subset $C\subseteq V $ satisfying the two requirements is called an \emph{identifying code} (any nonempty subset of $V$ is called a \emph{code}). The concept of identifying codes was introduced in \cite{KCL98}, where a fault diagnosis was performed in multiprocessor networks.

\smallskip

Consider a graph with the vertex set $\mathbb{Z}^2$ endowed with the
Euclidean distance $d$. Let further $r$ be a positive real number.
The set of edges is defined as follows: there is an edge between
two vertices of $\Z^2$ if their Euclidean distance is at most $r$. In other words, the closed neighbourhood of a vertex $u=(x_u,y_u)\in \mathbb{Z} ^2$ is the \emph{ball} 
$$ B_r(u)=\{v\in \mathbb{Z}^2\mid d(u,v)\leq r\}=\{(x,y)\in \mathbb{Z}^2\mid (x-x_u)^2+(y-y_u)^2\leq r^2\}.$$
Hence, in this graph each sensor can check (or cover) vertices within Euclidean distance $r$. In Figure~\ref{fig:euclideancode}, we have illustrated the graph with the ball of radius $\sqrt{5}$ and an identifying code in it. For special values of $r$, this graph has been considered in many papers, for example, \cite{BL05,CHHL04,HL03,HL02,JL11}. For related results, see \cite{CIQC02,MS09}. An interested reader is also referred to the web page \cite{lowww}, which contains an extensive collection of papers concerning identification and related problems.


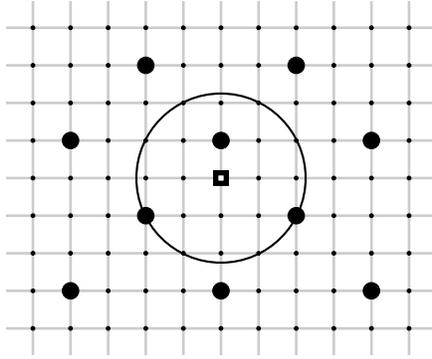
\begin{figure}
\begin{center}
\begin{tikzpicture}[scale=0.5]

\clip (-5.7,-4.7) rectangle (5.7,4.7);
\draw[mygrid] (-6.9,-5.9) grid (6.9,5.9);
\foreach \I in {-6,-5,-4,-3,-2,-1,0,1,2,3,4,5,6}\foreach \J in {-5,-4,-3,-2,-1,0,1,2,3,4,5}
	{\node[gridnode](\I\J) at (\I,\J) {};}

\draw[ball] (0,0) circle (2.25cm);
\node[centernode] at (0,0) {};


\foreach \I in {-6,-2,2,6} \foreach \J in {-5,-1,3}
	{\node[code] at (\I,\J) {};}
	
\foreach \I in {-4,0,4} \foreach \J in {-3,1,5}
	{\node[code] at (\I,\J) {};}

\end{tikzpicture}
\end{center}
\caption{\label{fig:euclideancode} For a Euclidean radius $r=\sqrt{5}$, the shaded vertices form an optimal identifying code for $\Z^2$. The sensors covering the squared vertex are illustrated by the Euclidean ball.}
\end{figure}

Since our underlying graph $\mathbb{Z}^2$ is infinite, we need a device to measure how `small' our code is compared to others. To this end, we use the usual concept of \emph{density}. Denote $Q_n=\{(x,y)\in \mathbb{Z}^2\mid |x|\le n, |y|\le n\}$ for a positive integer $n$. Obviously, $|Q_n|=(2n+1)^2.$ The density $D(C)$ of a code $C\subseteq \mathbb{Z}^2$ is $$D(C)=\limsup_{n\to +\infty} \frac{|C\cap Q_n|}{|Q_n|}.$$ We say that an identifying code is {\em optimal} if there is no code with lower density.


The previously defined concept of regular identification is a somewhat idealized view to approach the described locating problem. In particular, it is unrealistic to assume that each sensor $c \in C$ monitors exactly the Euclidean ball $B_r(c)$ of radius $r$. In this paper, we concentrate on a more realistic scenario, where the area that each sensor monitors may individually vary. Supposing $\Delta$ is a non-negative real number, we assume that the area covered by a sensor $c \in C$ is a subset of $B_{r+\Delta}(c)$ such that all the vertices of $B_r(c)$ belong to it. Consider then which sensors monitor a given vertex $u \in \mathbb{Z}^2$. Clearly, the sensors covering $u$ are the ones that belong to an area, which is a subset of $B_{r+\Delta}(u)$ and contains all the vertices of $B_r(u)$. A code $C$, using which we can uniquely determine the sought vertex (if any) solely based on the information provided by the sensors, is called $(r, \Delta)$\emph{-tolerant identifying} (or in short $(r, \Delta)$\emph{-identifying}). The formal definition of tolerant identifying codes is presented in the following.


Let us denote for $u$ and $v$ in $\Z^2$:
$$\S_{r,\Delta}(u,v)=\left(B_r(u)\setminus B_{r+\Delta}(v)\right)\cup
\left(B_r(v)\setminus B_{r+\Delta}(u)\right).$$
A subset $C\subseteq \Z^2$ is an $(r,\Delta)$-tolerant identifying code (or in short $(r,\Delta)$-identifying code), if for every $u\in \Z^2$ we have $B_r(u)\cap C\neq \emptyset$ and for all distinct vertices $u$ and $v$:
$$\S_{r,\Delta}(u,v) \cap C \neq \emptyset.$$
Clearly, this formal definition coincides with the informal one described above.
We denote the smallest possible density $D(C)$ of an $(r,\Delta)$-identifying code $C$ by $D(r,\Delta)$. An $(r,0)$-identifying code is simply an identifying code in the graph with vertex set $\Z^2$ and Euclidean radius $r$, studied in \cite{JL11}. For other results on robustness for identifying codes, see \cite{HKL06,HL07,RUPTS03,S02}. For comparing the sizes of regular $r$-identifying codes and $(r, \Delta)$-identifying codes with $\Delta > 0$ for small $r$, we refer to Table~\ref{table:small}. \smallskip

The following result is useful, when we bound the density of a code
from below. 
For $S\subseteq \mathbb{Z}^2$ and $v\in \mathbb{Z}^2$ we
denote a translate $v+S=\{v+s\mid s\in S\}$.

\begin{proposition}[\cite{JL11}]\label{prop:lowpattern} Let $S$ be a set of $k$ different
vertices of $\mathbb{Z}^2$. If a code $C\subseteq \mathbb{Z}^2$ is
such that  $|(v+S)\cap C|\ge m$ for all $v\in \mathbb{Z}^2$, then
$$D(C)\ge \frac{m}{k}.$$
\end{proposition}
In the sequel, we denote by $\C_r(u)=\{v\in \mathbb{R}^2 \  | \ d(u,v)=r\}$ the circle centered at $u\in \mathbb{Z}^2$ with radius $r$.

We start the paper in Section~\ref{SectionFirstExample} with a simple example of tolerant identifying codes with Euclidean balls. Then, in Sections~\ref{SectionGeneral1} and \ref{sec:modulo}, we present lower and upper bounds for $(r,\Delta)$-identifying codes for general values of $r$ and $\Delta$. In particular, we obtain infinite families of optimal codes (see Section~\ref{sec:goodcode}). Finally, in Section~\ref{SectionSmall}, we end the paper by considering codes with some specific (small) values of $r$ and $\Delta$.

%
%
%
%
%
%
%
%
%
%
%
\section{First example} \label{SectionFirstExample}

The first nontrivial ball of $\Z^2$ has radius $1$ and the next one has radius $\sqrt{2}$. In this section, we study this first nontrivial case where $r=1$ and $r+\Delta=\sqrt{2}$.
The pairs of vertices that are the 'most' difficult to separate (in the sense that there are few vertices that distinguish them) are pairs of vertices at distance $1$. Also it is good to start by studying $\S_{1,\sqrt{2}-1}(u,v)$ when $v-u=(1,0)$. We will call this set the {\em horizontal pattern}, and for the precise case $(r,\Delta)=(1,\sqrt{2}-1)$ this set is shown in Figure~\ref{fig:basic}. One can notice  that it has only two elements, that means, because of Proposition~\ref{prop:lowpattern}, that the density of an $(1,\sqrt{2}-1)$-identifying code is at least $\frac{1}{2}$. By symmetry, we also know the {\em vertical pattern} depicted in Figure~\ref{fig:basicvert}.

 \begin{figure}[h]
    \begin{center}
      \subfloat[][ \label{fig:basic}]{
      \begin{tikzpicture}[scale=0.6, baseline =-40]
\draw[mygrid] (-2.9,-0.9) grid (1.9,0.9);
\foreach \I in {-2,...,1}\foreach \J in {0}
	{\node[gridnode](\I\J) at (\I,\J) {};}
\foreach \pos in {(-2, 0), (1, 0)}
	\node[pattern] at \pos {};
	
\node[centernode] at (0,0) {};
\node[labelnode] at (0,0.5) {$v$};
\node[centernode] at (-1,0) {};
\node[labelnode] at (-1,0.5) {$u$};
\end{tikzpicture}
      }\hfil
        \subfloat[][ \label{fig:basicvert}]{
      \begin{tikzpicture}[scale=0.6]
\draw[mygrid] (-0.9,-2.9) grid (0.9,1.9);
\foreach \I in {-2,...,1}\foreach \J in {0}
	{\node[gridnode](\I\J) at (\J,\I) {};}
\foreach \pos in {(0,-2), (0, 1)}
	\node[pattern] at \pos {};
	
\node[centernode] at (0,0) {};
\node[labelnode] at (0.5,0) {$v$};
\node[centernode] at (0,-1) {};
\node[labelnode] at (0.5,-1) {$u$};
\end{tikzpicture}
      }\hfil
      \subfloat[][\label{fig:basicdiag}]{
      \begin{tikzpicture}[scale=0.6]
\draw[mygrid] (-2.9,-2.9) grid (1.9,1.9);
\foreach \I in {-2,...,1}\foreach \J in {-2,...,1}
	{\node[gridnode](\I\J) at (\I,\J) {};}
\foreach \pos in {(-2, -1), (-1,-2),(1, 0),(0,1)}
	\node[pattern] at \pos {};
	
\node[centernode] at (0,0) {};
\node[labelnode] at (0,-0.5) {$v$};
\node[centernode] at (-1,-1) {};
\node[labelnode] at (-1,-0.5) {$u$};
\end{tikzpicture}}
\hfil
      \subfloat[][\label{fig:basicantidiag}]{
      \begin{tikzpicture}[scale=0.6]
\draw[mygrid] (-2.9,-2.9) grid (1.9,1.9);
\foreach \I in {-2,...,1}\foreach \J in {-2,...,1}
	{\node[gridnode](\I\J) at (\I,\J) {};}
\foreach \pos in {(0,-2), (-1,1),(1, -1),(-2,0)}
	\node[pattern] at \pos {};
	
\node[centernode] at (-1,0) {};
\node[labelnode] at (0,-0.5) {$u$};
\node[centernode] at (0,-1) {};
\node[labelnode] at (-1,-0.5) {$v$};
\end{tikzpicture}
      }
      \caption{The set $\S_{1,\sqrt{2}-1}(u,v)$ when (a) $v-u=(1,0)$ (horizontal pattern), (b) $v-u=(0,1)$ (vertical pattern),  (c)  $v-u=(1,1)$ (diagonal pattern) and (d) $v-u=(1,-1)$ (anti-diagonal pattern).}
    \end{center}
  \end{figure}
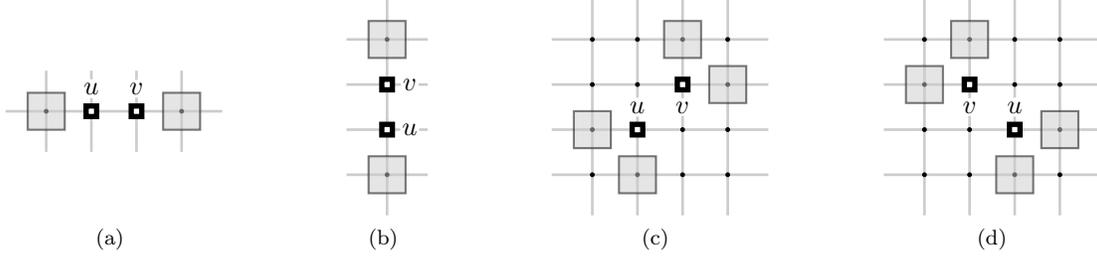

It is easy to find a code with density $\frac{1}{2}$ that is intersecting all the sets $\S_{1,\sqrt{2}-1}(u,v)$ for  $v-u=(1,0)$ or  $v-u=(0,1)$. One can for example take  as a code all the vertices $(x,y)$ such that $x+y\equiv 0 \bmod 2$. But this code will not always intersect the {\em diagonal pattern} $\S_{1,\sqrt{2}-1}(u,v)$ with $v-u=(1,1)$, that is shown on Figure~\ref{fig:basicdiag}.
If we take the code $C$ that is depicted on Figure~\ref{fig:basiccode}, we can show that it is always intersecting the diagonal pattern and the {\em anti-diagonal} pattern ($v-u=(-1,1)$).
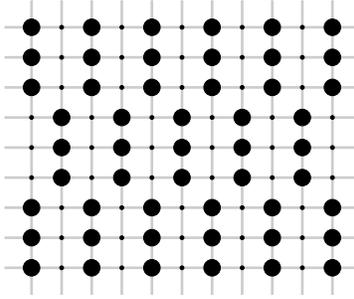
\begin{figure}[h]
\begin{center}
\begin{tikzpicture}[scale=0.4, rotate = 90]
\draw[mygrid] (-4.9,-5.9) grid (4.9,5.9);
\foreach \I in {-4,...,4}\foreach \J in {-5,...,5}
	{\node[gridnode](\I\J) at (\I,\J) {};}
	
\foreach \I in {-4,...,-2,2,3,4} \foreach \J in {-5,-3,...,5}
	{\node[code] at (\I,\J) {};}

\foreach \I in {-1,0,1} \foreach \J in {-4,-2,...,4}
	{\node[code] at (\I,\J) {};}
\end{tikzpicture}
\caption{\label{fig:basiccode} A $(1,\sqrt{2}-1)$-identifying code of optimal density $\frac{1}{2}$.}
\end{center}
\end{figure}
To show that $C$ is a $(1,\sqrt{2}-1)$-identifying code, it remains to show that it is a $1$-dominating set (this is clearly true) and  to check that it is intersecting all the other sets $\S_{1,\sqrt{2}-1}(u,v)$.  For this last point, we can notice that $\S_{1,\sqrt{2}-1}(u,v)$, when $d(u,v)>\sqrt{2}$, has always three vertices forming an $L$-pattern: \Lpat (up to orientation) and that $C$ is intersecting all the $L$-patterns. We can now conclude that $C$ has optimal density and so:
\begin{equation*}
 \D(1,\sqrt{2}-1)=\frac{1}{2} \textrm{.}
 \end{equation*}

\section{General Results} \label{SectionGeneral1}
\subsection{Existence of a code}

We give a necessary and sufficient condition to have an $(r,\Delta)$-identifying code for given values of $r$ and $\Delta$:

\begin{proposition}
There exists an $(r,\Delta)$-identifying code if and only if $\S_{r,\Delta}((0,0),(-1,0))$ is nonempty.
\end{proposition}

\begin{proof}
By definition, if there exists an $(r,\Delta)$-identifying code $C$, then there should be an element of $C$ in  $\S_{r,\Delta}((0,0),(-1,0))$. Therefore, this set is nonempty.

For the other side, assume that $\S_{r,\Delta}((0,0),(-1,0))$ is nonempty.  Then by symmetry, $\S_{r,\Delta}((0,0),(0,-1))$ is also nonempty as well as all the sets $\S_{r,\Delta}(u,v)$ with $d(u,v)=1$. Let $u,v$ be two vertices of $\Z^2$, with $v=u+(x,y)$. By symmetry, we can assume that $x>0$ and $y\geq 0$. Since $\S_{r,\Delta}((0,0),(-1,0))$ is nonempty, it necessarily contains a vertex $(x',y')$ with $x'>0$ and $y'\geq 0$, and then $v+(x',y')$ is in $\S_{r,\Delta}(u,v)$. Hence, all the sets  $\S_{r,\Delta}(u,v)$ are nonempty and the whole set $\Z^2$ is an $(r,\Delta)$-identifying code.
 \end{proof}

We consider the following definition:
$$\Delta_m(r)=\sup \{\Delta | \text{ there is an } (r,\Delta)\text{-identifying code}\}.$$

\noindent Note that if $\Delta\geq 1$, then clearly, for any radius $r$, $\S_{r,\Delta}((0,0),(-1,0))$ is empty and so there is no $(r,\Delta)$-identifying code. Hence, $\Delta_m(r)$ is well-defined and $\Delta_m(r)\leq 1$.
Furthermore, for a fixed $r$, an $(r,\Delta)$-identifying code for some $\Delta\geq 0$, is also an $(r,\Delta')$-identifying code for any $0\leq \Delta'\leq \Delta$. Therefore, for any $0\leq \Delta < \Delta_m(r)$ there is an $(r,\Delta)$-identifying code whereas for any $\Delta \geq \Delta_m(r)$ there is no $(r,\Delta)$-identifying code.
When $r$ is an integer, $\Delta_m(r)= 1$, because $(r,0) \in \S_{r,\Delta}((0,0),(-1,0))$ for $\Delta <1$.  When $r$ is not an integer, we have:

\begin{proposition}\label{prop:maxdelta}
When $r\to +\infty$, we have
$$\Delta_m(r) \geq 1-\sqrt{\frac{2}{r}}+O\left(\frac{1}{r}\right).$$
\end{proposition}

\begin{proof}
Consider the largest abscissa $\lfloor r \rfloor$ of a vertex of $B_r((0,0))$, and let $u$ be the vertex of $B_r((0,0))$ with the largest ordinate in the column $\lfloor r \rfloor$.
Then $u$ has coordinates $(\lfloor r \rfloor,  \lfloor \sqrt{r^2 - \lfloor r \rfloor^2} \rfloor)$.
Let $\alpha=r-\lfloor r \rfloor$, then we have $u=(r-\alpha,\lfloor \sqrt{\alpha(2r-\alpha)} \rfloor)$.
Denote $\Delta_0 = d((-1,0),u)-r$. Then, we have $\Delta_0>0$. For $\Delta < \Delta_0$, $u\in \S_{r,\Delta}((0,0),(-1,0))$, so $\Delta_m(r)\geq \Delta_0$. We now compute a lower bound of $\Delta_0$, using the fact that $\alpha \in [0,1[$:

\begin{eqnarray*}
r+\Delta_0 & = &  \left( (r-\alpha +1)^2+ (\lfloor \sqrt{\alpha(2r-\alpha)} \rfloor)^2\right)^{\frac{1}{2}}\\
&\geq & \left( (r-\alpha +1)^2+ (\sqrt{\alpha(2r-\alpha)}-1)^2\right)^{\frac{1}{2}}\\
&\geq & (r^2 + 2r - 2\sqrt{\alpha(2r-\alpha)} +2-2\alpha)^{\frac{1}{2}}\\
&\geq & (r^2+2r- 2\sqrt{2r})^{\frac{1}{2}}\\
&=& r\left(1+\frac{2}{r} -  \frac{2\sqrt{2}}{r^{3/2}}\right)^{\frac{1}{2}}\\
\end{eqnarray*}

As $r\to +\infty$, using Taylor series, we have:

\begin{eqnarray*}
r+\Delta_0  &\geq & r\left(1+\frac{1}{r} -  \frac{\sqrt{2}}{r\sqrt{r}}+ O\left(\frac{1}{r^2}\right)\right).\\
\end{eqnarray*}

And so:
\begin{equation*}
\Delta_m(r) \geq \Delta_0 \geq 1- \sqrt{\frac{2}{r}} + O\left(\frac{1}{r}\right).
\end{equation*}

\end{proof}

In the following, we will consider only values $(r,\Delta)$ such that there is an $(r,\Delta)$-identifying code.

\subsection{Study of the horizontal pattern $\S_{r,\Delta}((0,0),(-1,0))$}

As said before, the pair of vertices that are the most difficult to identify are vertices at distance $1$. Thus, it is important to have a good knowledge of the horizontal pattern $\S_{r,\Delta}((0,0),(-1,0))$ (by symmetry, that will also give us knowledge on the vertical pattern). In what follows we will use the following notations:

\begin{itemize}
\item $h_{r,\Delta}(x)=\sqrt{r^2-x^2} -\sqrt{(r+\Delta)^2-(x+1)^2}$, defined for $x\in [0,\lfloor r+\Delta-1 \rfloor]$ is the signed vertical distance at abscissa $x$ between the circles $\C_r((0,0))$ and $\C_{r+\Delta}((-1,0))$. When $x\geq 0$ and $h_{r,\Delta}(x)\leq 0$, there cannot be any vertex with abscissa $x$ in $\S_{r,\Delta}((0,0),(-1,0))$. When $h_{r,\Delta}(x)> 1$, there is always a vertex with abscissa $x$ in $\S_{r,\Delta}((0,0),(-1,0))$. Finally, when $h_{r,\Delta}(x)\in ]0,1]$, there is at most one vertex with abscissa $x$ in $\S_{r,\Delta}((0,0),(-1,0))$.
That leads to the next definitions:

\item $x_0(r,\Delta)$ (or $x_0$ when the context is clear) is the smallest nonnegative abscissa of an element of $\S_{r,\Delta}((0,0),(-1,0))$. It is at least the ceiling of the positive solution of $h_{r,\Delta}(x)=0$. We have:
\begin{equation*}
x_0(r,\Delta) \geq \left \lceil {\frac{\Delta(2r+\Delta)-1}{2}}\right \rceil = r\Delta +O(1) \text{ as } r\to +\infty.
\end{equation*}

\item $x_1(r,\Delta)$ (or $x_1$ when the context is clear) is the floor of the positive solution of $h_{r,\Delta}(x)=1$. The exact value of $x_1$ is:
\begin{equation}\label{eq:x1}
x_1(r,\Delta) =\left \lfloor { \frac{1}{4}\left(-2+\Delta^2+2r\Delta+ \sqrt{-4+4\Delta^2-\Delta^4+8r\Delta -4r\Delta^3+8r^2-4r^2\Delta^2}\right)}\right \rfloor.
\end{equation}

As $r\to +\infty$, we obtain:
\begin{equation*}
x_1(r,\Delta) = r\left(\frac{\Delta}{2}+\frac{\sqrt{2-\Delta^2}}{2}\right) + O(1).
\end{equation*}

We will also need the value of $x_1$ when $\Delta$ is close to $1$. Assume that $\Delta =1 -\epsilon(r)$ with $\epsilon(r)$ of order $\frac{1}{r^{\alpha}}$ for some real number $\alpha$ as $r\to +\infty$. By Proposition~\ref{prop:maxdelta}, we always have an $(r,\Delta)$-identifying code if $\epsilon(r)\geq \sqrt{\frac{2}{r}}+O(\frac{1}{r})$. Hence, we assume that $0<\alpha\leq \frac{1}{2}$. Then, as $r\to +\infty$:

\begin{equation}\label{eq:x1eps}
x_1(r,\Delta) = r-\frac{r\epsilon(r)^2}{2} + O(r\epsilon(r)^3) + O(1).
\end{equation}

In the special case $\alpha=\frac{1}{2}$, we obtain that $r-x_1$ is bounded.

\item $m(r,\Delta)$ (or $m$ when the context is clear) is the ceiling of the value of $h_{r,\Delta}$ in abscissa $\lfloor r+\Delta -1 \rfloor$. The exact value of $m$ is:
\begin{equation}\label{eq:m}
m(r,\Delta)= \left\lceil \sqrt{r^2-\lfloor r+\Delta -1 \rfloor^2}-\sqrt{(r+\Delta)^2-\lfloor r+\Delta\rfloor^2} \right\rceil.
\end{equation}

We have, when $r\to +\infty$:

\begin{equation} \label{eq:mApprox}
m(r,\Delta)\leq 2\sqrt{r}+O(1) \textrm{.}
\end{equation}

\end{itemize}

\noindent Because $h_{r,\Delta}$ is a strictly increasing function on $[0,\lfloor r+\Delta-1 \rfloor]$, $m(r,\Delta)$ is an upper bound of $h_{r,\Delta}$. Furthermore, there are, in $\S_{r,\Delta}((0,0),(-1,0))$, for an integer $x\in[0,\lfloor r+\Delta-1 \rfloor]$:
\begin{itemize}
\item no vertex of abscissa $x$ if $x< x_0$,
\item at most one vertex of abscissa $x$ if $x_0 \leq x\leq x_1$,
\item at least one and at most $m$ vertices of abscissa $x$ if $x>x_1$.
\end{itemize}

We finish this section by a result that will be used in Section~\ref{sec:modulo}:

\begin{lemma}\label{lem:dist}
Let $u$ and $v$ be two vertices of $\Z^2$, lying on the same vertical or horizontal line.
If $d(u,v)\leq 4x_0(r,\Delta)+1$, then there are two vertices $u'$ and $v'$ of $\Z^2$, at distance $1$, such that $\S_{r,\Delta}(u',v')\subseteq \S_{r,\Delta}(u,v)$.
\end{lemma}

\begin{proof}
We assume first that $d(u,v)$ is odd: $d(u,v)=2k+1$ with $k\leq 2x_0$.
Without loss of generality, we can assume that $u=(-k-1,0)$ and that $v=(k,0)$.
Let $u'=(-1,0)$ and $v'=(0,0)$.
We will show that $\S_{r,\Delta}(u',v')\subseteq \S_{r,\Delta}(u,v)$.
Let $w=(x,y) \in \S_{r,\Delta}(u',v')$. By symmetry we can assume that $w\in B_r(v')\setminus B_{r+\Delta}(u')$, and that $y\geq 0$. Then clearly $d(w,u)\geq d(w,u')>r+\Delta$ so $w\notin B_{r+\Delta}(u)$.
If $x\geq k$, then $d(w,v)\leq d(w,v')\leq r$. Otherwise, $x\geq x_0$ and $k\leq 2x_0$, so $k-x\leq x$, and we also have $d(w,v)\leq d(w,v')$.
Hence, finally $w\in B_r(v)\setminus B_{r+\Delta}(u)\subseteq \S_{r,\Delta}(u,v)$.
The case $d(u,v)$ even is the same.
\end{proof}

\subsection{Lower bound}

The bound used in \cite{JL11} when $\Delta=0$ is still valid here:

\begin{equation}\label{eq:lowerfixed}
D(r,\Delta)\geq D(r,0)\geq \frac{1}{3.22r+4}.
\end{equation}

This bound is good for fixed value of $\Delta$ when $r$ grows. But when $\Delta$ approaches $1$ as $r$ tends to infinity, we will have a better bound. For this we will use the following proposition. This lower bound is sharp in some cases as we will see in the next section. It is a direct consequence of Proposition~\ref{prop:lowpattern}:

\begin{proposition}\label{prop:lowerbound}
We have
\[
\D(r,\Delta)\geq \frac{1}{\vert \S_{r,\Delta}((0,0),(-1,0)) \vert} \textrm{.}
\]
\end{proposition}

One can write the exact value of $\vert \S_{r,\Delta}((0,0),(-1,0)) \vert$:
\begin{equation*}
\vert \S_{r,\Delta}((0,0),(-1,0)) \vert = 4\sum_{x=x_0(r,\Delta)}^{\lfloor r +\Delta -1 \rfloor}\left(\lfloor\sqrt{r^2-x^2}\rfloor- \lfloor\sqrt{(r+\Delta)^2-(x+1)^2}\rfloor\right)
+ \delta \cdot (4\lfloor\sqrt{r^2-\lfloor r \rfloor^2}\rfloor +2)
\end{equation*}
\noindent where $\delta=1$ if $\lfloor r \rfloor =\lfloor r+\Delta\rfloor$ and $0$ otherwise.

To obtain a concrete  lower bound of $\D(r,\Delta)$, we estimate $|\S_{r,\Delta}((0,0),(-1,0))|$ using previous notation and noticing that $\lfloor\sqrt{r^2-x^2}\rfloor- \lfloor\sqrt{(r+\Delta)^2-(x+1)^2}\rfloor$ is between $\lfloor h_{r,\Delta}(x) \rfloor$ and $\lceil h_{r,\Delta}(x) \rceil$:

\begin{eqnarray*}
\vert \S_{r,\Delta}((0,0),(-1,0)) \vert &\leq& 4((x_1-x_0+1) +m\cdot(\lfloor r+ \Delta -1\rfloor-x_1))+ \delta \cdot (4\lfloor\sqrt{r^2-\lfloor r \rfloor^2}\rfloor +2) \textrm{.}\\
\end{eqnarray*}

Assume that $\Delta =1 -\epsilon(r)$ with $\epsilon(r)$ of order $\frac{1}{r^{\alpha}}$ as $r\to \infty$, with $\alpha\leq \frac{1}{2}$. Using the equations~\eqref{eq:x1eps} and \eqref{eq:mApprox}, we can show that the order of $\vert \S_{r,\Delta}((0,0),(-1,0)) \vert$ is at most $r\sqrt{r}\epsilon(r)^2$ and so the density  $D(r,\Delta)$ has order at least $\frac{1}{r\sqrt{r}\epsilon(r)^2}$.


In the particular case where $\alpha=\frac{1}{2}$, i.e., $\Delta$ is really close to $\Delta_m$, we obtain that $D(r,\Delta)$ has order at least $\frac{1}{\sqrt{r}}$ which is better than the order in the general case.


\subsection{Upper bound}\label{sec:consgeneral}

To obtain upper bounds, we construct codes with the following basic sets, defined for a positive integer $k$:
\begin{itemize}
\item $L_k^{v}=\{(x,y) \in \Z^2 | \ x\equiv 0\bmod k\}$,
\item $L_k^{h}=\{(x,y) \in \Z^2 | \ y \equiv 0\bmod k\}$.
\end{itemize}

\begin{proposition}
Let $k=\lfloor r\rfloor-x_1(r,\Delta)$. Then $C_{r,\Delta}=L_k^v\cup L_k^h$ is an $(r,\Delta)$-identifying code, and we have:
\begin{equation*}
D(r,\Delta)\leq \frac{2}{k} \textrm{.}
\end{equation*}
\end{proposition}

\begin{proof}

Let $X_1=\S_{r,\Delta}((0,0),(-1,0))\cap \{(x,y) \in \Z^2 | \ x\geq 0, y\geq 0\}$ be the vertices of $\S_{r,\Delta}((0,0),(-1,0))$ that lie in the first quadrant. In the same way, let $X_2 = \S_{r,\Delta}((0,0),(-1,0))\cap \{ (x,y)  \in \Z^2| \ x\geq 0, y\leq 0\}$, and for the vertical pattern: $X_3=\S_{r,\Delta}((0,0),(0,-1))\cap \{(x,y)  \in \Z^2| \ x\geq 0, y\geq 0\}$ and $X_4=\S_{r,\Delta}((0,0),(0,-1))\cap \{(x,y)  \in \Z^2| \ x\leq 0, y\geq 0\}$.
For any pair of vertices $u$ and $v$ of $\Z^2$, with $u\neq v$, there is $i\in \{1,2,3,4\}$ such that $u+X_i\subset \S_{r,\Delta}(u,v)$ or $v+X_i\subset \S_{r,\Delta}(u,v)$ .

In $X_1$, by definition of $x_1(r,\Delta)$, there are vertices with abscissas between $x_1+1$ and $\lfloor r \rfloor$, so there are vertices with $k$ consecutive abscissas in $X_1$.
This implies that $(u+X_1)\cap C_{r,\Delta} \neq \emptyset$ for all $u\in \Z^2$.
This is also true for $X_2$, $X_3$, $X_4$ (by considering the ordinates and horizontal lines for $X_3$ and $X_4$).
Finally, for any pair of vertices $u$ and $v$,  $\S_{r,\Delta}(u,v)\cap C_{r,\Delta} \neq \emptyset$. Moreover, $C_{r,\Delta}$ is an $r$-dominating set and so it is an $(r,\Delta)$-identifying code.
Clearly, $C_{r,\Delta}$ has density $\frac{2}{k}$.
\end{proof}

As $r$ grows, this leads to the following upper bound:
\begin{equation}\label{eq:upperfixed}
D(r,\Delta) \leq \frac{4}{r(2-\Delta-\sqrt{2-\Delta^2})+K}
\end{equation}
where $K$ is a constant. Combining (\ref{eq:lowerfixed}) and (\ref{eq:upperfixed}), we know that an optimal $(r,\Delta)$-identifying code, for $\Delta$ fixed, has order $\frac{1}{r}$ as $r\to +\infty$.

Assume that $\Delta=1-\epsilon(r)$, with $\epsilon(r)$ of order $\frac{1}{r^{\alpha}}$ as $r\to \infty$, and $\alpha\leq \frac{1}{2}$.  Then, we obtain that $D(r,\Delta)$ has order at most $\frac{1}{r\epsilon(r)^2}$. In the particular case where $\Delta$ is really close to $\Delta_m$ ($\alpha=\frac{1}{2}$), the result is trivial. As we will see in the next section, for infinite family of $(r,\Delta)$ there are optimal codes of density $\frac{1}{2}$ so we cannot expect in this case to have a general upper bound of order better than a constant.

\section{Better constructions for given values of $(r,\Delta)$}\label{sec:modulo}
\subsection{General construction}
The construction of Section~\ref{sec:consgeneral} does not use the full symmetry of the set  $\S_{r,\Delta}((0,0),(-1,0))$. We can often construct better codes using diagonal lines that utilize the symmetry. For this we need to have more information about $\S_{r,\Delta}((0,0),(-1,0))$ and $\S_{r,\Delta}((0,0),(-1,-1))$ that we cannot compute in the general case. The following construction of Proposition~\ref{prop:method} should be seen as a method to construct $(r,\Delta)$-identifying codes for a given value of $(r,\Delta)$.

 We say that a set $U$ of $\Z^2$ is {\em intersecting all the diagonal} (resp. {\em anti-diagonal, horizontal and vertical}) {\em lines modulo $k$}  if for all $i\in\{0,1,\ldots,k-1\}$, there is an element $(x,y)$ of $U$ such that $y-x\equiv i \bmod k$ (resp. $x+y\equiv i\bmod k$, $y\equiv i\bmod k$ and $x \equiv i \bmod k$). As an example, the set $\S_{3,\sqrt{10}-3}((0,0),(-1,0))$ (see Figure~\ref{fig:modulo}) is intersecting all the diagonal and anti-diagonal lines modulo 6 whereas the set $\S_{3,\sqrt{10}-3}((0,0),(-1,-1))$ is intersecting all the horizontal and vertical lines modulo 8.

 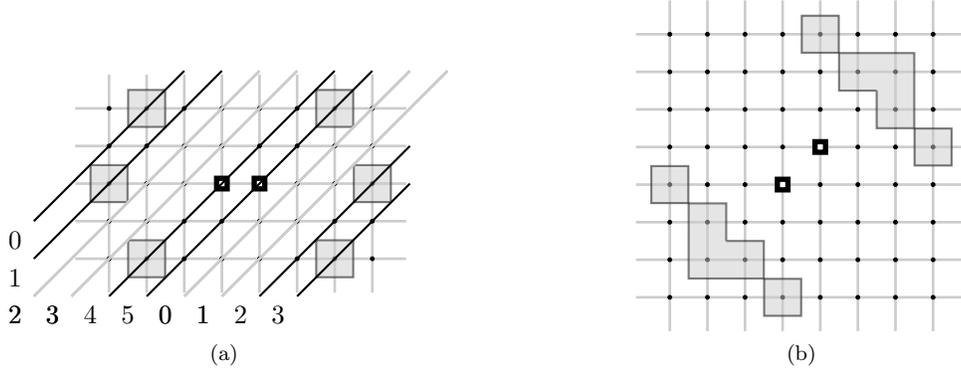
\begin{figure}[h]
 \begin{center}
 \subfloat[][\label{fig:diagmod6}]{
 \begin{tikzpicture}[scale=0.5]
\draw[mygrid] (-4.9,-2.9) grid (3.9,2.9);
\foreach \I in {-4,...,3}\foreach \J in {-2,...,2}
	{\node[gridnode](\I\J) at (\I,\J) {};}
\foreach \pos in {(-3, -2),(-4,0),(-3,2), (2,-2),(2,2),(3,0)}
	\node[pattern] at \pos {};
	
\node[centernode] at (0,0) {};
\node[centernode] at (-1,0) {};

\draw[thick, xshift=-3 cm] (-3,-1)--(1,3);
\node at (-3.5-3,-1.5) {\color{roug}$0$};
\draw[thick, xshift=-2 cm] (-4,-2)--(1,3);
\node at (-4.5-2,-2.5) {\color{roug}$1$};

\foreach \I / \d in {1/4,2/5}
{\draw[thick, xshift=\I cm] (-5,-3)--(1,3);
\node at (-5.5+\I,-3.5) {\color{roug} $\d$};
}

\foreach \I / \d in {0/3,-1/2,3/0,4/1}{
\draw[mygrid,xshift=\I cm] (-5,-3)--(1,3);
\node at (-5.5+\I,-3.5) {$\d$};
\draw[mygrid,xshift=\I cm] (-5,-3)--(1,3);
\node at (-5.5+\I,-3.5) {$\d$};
}
\draw[thick, xshift=5 cm] (-5,-3)--(-1,1);
\node at (-5.5+5,-3.5) {\color{roug}$2$};
\draw[thick, xshift=6 cm] (-5,-3)--(-2,0);
\node at (-5.5+6,-3.5) {\color{roug}$3$};
\end{tikzpicture}
}\hfil
\subfloat[][\label{fig:hormod8}]{
\begin{tikzpicture}[scale=0.5]
\draw[mygrid] (-4.9,-4.9) grid (3.9,3.9);
\foreach \I in {-4,...,3}\foreach \J in {-4,...,3}
	{\node[gridnode](\I\J) at (\I,\J) {};}
\foreach \pos in {(-1,-4),(-4,-1),(0,3),(3,0)}
	\node[pattern] at \pos {};
	
\node[centernode] at (0,0) {};
\node[centernode] at (-1,-1) {};

\draw[patternpath, shift = {(1,1)}] (0.5,0.5)--(-0.5,0.5)--(-0.5,1.5)--(1.5,1.5)--(1.5,-0.5)--(0.5,-0.5)--(0.5,0.5);
\draw[patternpath, shift = {(-2,-2)}, rotate = 180] (0.5,0.5)--(-0.5,0.5)--(-0.5,1.5)--(1.5,1.5)--(1.5,-0.5)--(0.5,-0.5)--(0.5,0.5);
\end{tikzpicture}
 }
 \caption{\label{fig:modulo} In (a), the set $\S_{3,\sqrt{10}-3}((0,0),(-1,0))$ is intersecting all the diagonal and anti-diagonal lines modulo 6. In (b), the set $\S_{3,\sqrt{10}-3}((0,0),(-1,-1))$ is intersecting all the horizontal and vertical lines modulo 8.}
 \end{center}
 \end{figure}

 Let $diag(U)$ be the maximum $k$ such that the set $U$ is intersecting all the diagonal lines modulo $k$.

 \begin{lemma}\label{lem:ball}
Let $u$ be a vertex in $\Z^2$. Then we have
\[ diag(B_r(u)) = 4\left \lfloor \frac{r}{\sqrt{2}} \right\rfloor+2\delta +1 \textrm{,}\]
where $\delta = 0$ if ${\left \lfloor \frac{r}{\sqrt{2}} \right \rfloor}^2 + \left(\left\lfloor\frac{r}{\sqrt{2}}\right\rfloor +1\right)^2 \leq r^2$ and $\delta = 1$ otherwise.
 \end{lemma}

\begin{proof}
It is enough to show the result for $u=(0,0)$.
Let $v=\left(- \left \lfloor \frac{r}{\sqrt{2}}\right \rfloor,\left \lfloor \frac{r}{\sqrt{2}} \right \rfloor\right)$ and
$w=\left(\left \lfloor \frac{r}{\sqrt{2}} \right \rfloor,-\left \lfloor \frac{r}{\sqrt{2}} \right \rfloor\right)$.
Then $v,w\in B_r(u)$ and all the diagonal lines between $v$ and $w$ (included $v$ and $w$) are intersecting $B_r(u)$.
There are $4\left \lfloor \frac{r}{\sqrt{2}} \right\rfloor+1$ such diagonal lines.
There can be one more diagonal before $v$, if the vertex $v+(-1,0)$ is in $B_r(u)$. That corresponds to the condition $\left \lfloor \frac{r}{\sqrt{2}} \right \rfloor ^2 + \left(\left \lfloor\frac{r}{\sqrt{2}}\right \rfloor +1\right) ^2 \leq r^2$ and in this case there is also one more diagonal after $w$.
\end{proof}

Note that the set $B_r((0,0))$ is intersecting all the diagonal lines modulo $s$, for $s\leq diag(B_r((0,0)))$. This is not always the case for a disconnected set: the set $\S_{3,\sqrt{10}-3}((0,0),(-1,0))$ is not intersecting all the diagonal lines modulo 4 but $diag(\S_{3,\sqrt{10}-3}((0,0),(-1,0))=6$ (see Figure~\ref{fig:modulo}).

  \begin{lemma}\label{lem:line}
If $\S_{r,\Delta}((0,0),(-1,0))$ is intersecting all the diagonal lines modulo $s$,  with $s\leq  diag(B_r((0,0)))$, then for any pair of vertices $u,v$ in $\Z^2$, that lie on the same horizontal or vertical line, $\S_{r,\Delta}(u,v)$ is intersecting all the diagonal lines modulo $s$.
 \end{lemma}

 \begin{proof}
Notice first that, if $u$ and $v$ lie on the same horizontal line and if $\S_{r,\Delta}(u,v)$ is intersecting all the diagonal lines modulo $s$, then it is also intersecting all the anti-diagonal lines modulo $s$ (because of the vertical symmetry of $\S_{r,\Delta}(u,v)$). Hence, $\S_{r,\Delta}(u',v')$, where $u'$ and $v'$ are the images of $u$ and $v$ by a rotation of $\frac{\pi}{2}$ centered in $(0,0)$, is also intersecting all the diagonal and anti-diagonal lines modulo $s$. Therefore, we just need to prove the proposition when the vertices lie on the same horizontal line.

Let us assume that $\S_{r,\Delta}((0,0),(-1,0))$ is intersecting all the diagonal lines modulo $s$,  with $s\leq  diag(B_r((0,0)))$.
Let $u,v$ be lying on the same horizontal line. If $d(u,v)\leq 4x_0+1$, by Lemma~\ref{lem:dist}, $\S_{r,\Delta}(u,v)$ is containing a set isomorphic to $\S_{r,\Delta}((0,0),(-1,0))$, and so is intersecting all the diagonal lines modulo $s$. If $d(u,v)> \lfloor r\rfloor +\lfloor r+\Delta \rfloor$, then $B_r(v)\setminus B_{r+\Delta}(u)=B_r(v)$. The set $B_r(v)$ is intersecting all diagonal lines modulo $s\leq diag(B_r(v))=diag(B_r((0,0)))$ and so is the set $\S_{r,\Delta}(u,v)$.

Therefore, we can now assume that $4x_0+2\leq d(u,v)\leq \lfloor r\rfloor +\lfloor r+\Delta \rfloor$. Without loss of generality, we assume that $u=(-k',0)$ and $v=(k,0)$ with $k=\left\lfloor\frac{d(u,v)}{2}\right\rfloor$, $k+k'=d(u,v)$ and $k'\geq k$. We have $ \Delta(2r+\Delta) \leq 2x_0+1 \leq k \leq \lfloor r \rfloor$. Let $w_u$ be the vertex $u+\left(-\left \lfloor \frac{r}{\sqrt{2}} \right \rfloor,\left \lfloor \frac{r}{\sqrt{2}} \right \rfloor\right)$ and $w_v$ be the vertex $v+\left(\left \lfloor \frac{r}{\sqrt{2}} \right \rfloor,-\left \lfloor \frac{r}{\sqrt{2}} \right \rfloor\right)$. We will show that $\S_{r,\Delta}(u,v)$ is intersecting all the diagonal lines between $w_u$ and $w_v$, which is more than $s$ consecutive diagonal lines.

Consider the diagonal line $D$ defined by $y=x$. We want to prove that all the diagonal lines between $D$ and $w_v$ are intersecting $B_r(v)\setminus B_{r+\Delta}(u)$. Since $k\leq \lfloor r \rfloor$, all those diagonal lines are intersecting $B_r(v)$ in a vertex of $\Z^2$ with positive abscissa. If $D$ is not intersecting $B_{r+\Delta}(u)$ with a nonnegative abscissa, then we are done. Hence, we can assume that $D$ is intersecting the circle $\C_{r+\Delta}(u)$ in a point with nonnegative abscissa $x_u$ (not necessarily an integer point). Let $x_v$ be the nonnegative abscissa of the intersection between $D$ and the circle $\C_r(v)$. Then if the distance $d$ between the two points $(x_u,x_u)$ and $(x_v,x_v)$  is more than $\sqrt{2}$, we are sure that there is at least one integer point in $D$ that lies in $B_r(v)\setminus B_{r+\Delta}(u)$. Since the distance between the two circles on each diagonal on the right of $D$ will be also greater than $\sqrt{2}$, there will an integer
  point for all the diagonal lines between $D$ and $w_v$ in $B_r(v)\setminus B_{r+\Delta}(u)$.

We know that the distance $d$ is $\sqrt{2}(x_v-x_u)$. Hence, we need to show that $x_v-x_u\geq 1$. We have:

\begin{eqnarray*}
x_v-x_u& = & \frac{k+k'}{2} + \frac{1}{2}(\sqrt{2r^2-k^2}- \sqrt{2r^2-k'^2+4r\Delta+2\Delta^2})\\
\end{eqnarray*}

Either $x_0=0$ and then $\Delta=0$,$k=k'=1$ and we are done, or $x_0\geq 1$ and then, $k\geq 3$ and:
\begin{eqnarray*}
x_v-x_u& \geq & k +\frac{1}{2}(\sqrt{2r^2-k^2}- \sqrt{2r^2-k^2+2k})\\
&=& k +\frac{2r^2-k^2-(2r^2-k^2+2k)}{2(\sqrt{2r^2-k^2}+ \sqrt{2r^2-k^2+2k})}\\
& \geq & k -\frac{k}{2\sqrt{2r^2-k^2}}\\
& \geq & \frac{k}{2} > 1.\\
\end{eqnarray*}

Now we can do the same for $B_r(u)\setminus B_{r+\Delta}(v)$ and then show that there exists a vertex in $B_r(u)\setminus B_{r+\Delta}(v)$ that lies in each diagonal between $w_u$ and the diagonal defined by $y=x+1$, completing the proof.
\end{proof}

The following proposition gives, for given $r$ and $\Delta$, a method to construct better $(r,\Delta)$-identifying codes than in the previous section. When $r$ and $\Delta$ are given, it is relatively easy to check with the patterns that the code constructed in the proposition is an $(r,\Delta)$-identifying code.  However, the general proof is very technical even for restricted values of $r$ and the outline of the proof is given in Appendix. Nevertheless, we believe that the method can be used for any value of $(r,\Delta)$.

\begin{proposition}\label{prop:method}
Let $s\leq diag(B_r((0,0)))$ and $t$ be integers. Let $L^d_s=\{(x,y) | y-x\equiv 0\bmod s\}$.
Assume that the following conditions hold:
\begin{enumerate}[(a)]
\item $\S_{r,\Delta}((0,0),(-1,0))$ is intersecting all the diagonal lines modulo $s$,
\item $\S_{r,\Delta}((0,0),(-1,-1))$ is intersecting all the horizontal lines modulo $t$,
\item $r$ is not too close to an integer :  $\sqrt{\lfloor r \rfloor^2+4}\leq r<\lfloor r \rfloor +1$.
\end{enumerate}

\noindent Then the code $C=L^d_s \cup L^h_t$ is an $(r,\Delta)$-identifying code of density $\frac{1}{s}+\frac{1}{t}-\frac{1}{st}$.
\end{proposition}

We can sometimes improve the code of the previous proposition by removing the vertices in the intersection between the horizontal and diagonal lines. The reason is that in the horizontal pattern, for any vertex there is another vertex in the same horizontal line. Hence, if the distance between any pair of vertices in the same line is not $s$, then if one vertex of the code is missing in a diagonal line, the horizontal pattern will intersect the code in the same horizontal line with the other vertex. The same holds for the diagonal pattern with the diagonal lines. As  an example, in the case $(r,\Delta)=(\sqrt{5},3-\sqrt{5})$, the horizontal pattern is intersecting all the diagonal lines modulo 4, the diagonal pattern is intersecting all the horizontal lines modulo $6$, and the code $C=(L^d_4 \cup L^h_6)\setminus (L^d_4 \cap L^h_6)$ is a $(\sqrt{5},3-\sqrt{5})$-identifying code of density $\frac{1}{3}$. This is not always working, as in the case $(r,\Delta)=(\sqrt{41},5\sqrt{2
 }-\sqrt{41})$.


\subsection{Optimal constructions} \label{sec:goodcode}
If in the diagonal pattern there are  all the diagonal and anti-diagonal lines modulo $s$, we do not need to put horizontal lines in the code of Proposition~\ref{prop:method}, and then we obtain a code of density $\frac{1}{s}$. This is in particular the case when $s$ is fixed and $r$ is large enough:

\begin{proposition}\label{prop:optimal} Let $s$ be a fixed integer.

{\bf (i)} There exists $r_0\in
\mathbb{N}$ such that for all $r\ge r_0$ and all $\Delta\in
[0,\Delta_m(r)]$ the set $\S_{r,\Delta}((0,0),(-1,-1))$ contains
all the diagonal and anti-diagonal lines modulo $s$ in the first
quadrant.

{\bf (ii)} If, furthermore, the set $\S_{r,\Delta}((0,0),(-1,0))$ contains all the diagonal lines modulo $s$, then there is an
$(r,\Delta)$-identifying code of density $\frac{1}{s}$.
\end{proposition}

\begin{proof} {\bf Claim (i):} Since $\Delta_m(r)\leq 1$,  everything that is
contained in $\S_{r,1}((0,0),(-1,-1))$ is also contained in
$\S_{r,\Delta}((0,0),(-1,-1))$. Therefore, for the first claim, we
assume that $\Delta=1$. We consider the vertex set $\mathbb{Z}^2$
being partitioned into two subsets (see Figure~\ref{fig:diaglattice}) of `even
vertices' $\mathbb{Z}^2_{e}=\{(i,j)\in \mathbb{Z}^2\mid i+j\equiv 0
\bmod 2\}$ and `odd vertices'
$\mathbb{Z}^2_{o}=\{(i,j)\in\mathbb{Z}^2\mid i+j\equiv 1\bmod 2\}$.
Clearly, $\mathbb{Z}^2_{e}$ can be considered as $\mathbb{Z}^2$
rotated (clockwise) by $\pi/4$ where  the unit length
between closest vertices being $\sqrt{2}$ (instead of $1$). Now the
diagonal lines (resp. anti-diagonal lines) of the original lattice
$\mathbb{Z}^2$ are vertical
 (resp. horizontal) lines of the new lattice $\mathbb{Z}^2_e$.

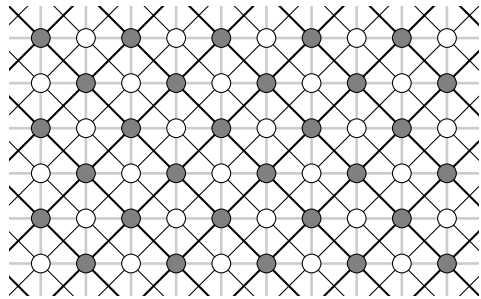
\begin{figure}[h]
\begin{center}
\begin{tikzpicture}[scale=0.6]

\clip (-3.7,-2.7) rectangle (6.7,3.7);

\draw[mygrid] (-6.9,-5.9) grid (6.9,5.9);

\draw[rotate = 45,thick] (-15.9,-15.9) grid [xstep=1.414,ystep =1.414] (15.9,15.9);
\draw[shift = {(1,0)}, rotate = 45] (-10.9,-10.9) grid [xstep=1.414,ystep =1.414] (10.9,10.9);

\foreach \I in {-6,-4,...,6}
	\foreach \J in {-5,-3,...,5}
	{\node[evennode] at (\I,\J) {};
	\node[oddnode] at (\I,\J-1) {};}

\foreach \I in {-5,-3,...,5}
	\foreach \J in {-5,-3,...,5}
	{\node[oddnode] at (\I,\J) {};
	\node[evennode] at (\I,\J-1) {};}
\end{tikzpicture}
\end{center}
\caption{\label{fig:diaglattice} Partition of $\Z^2$ in 'even' and 'odd vertices'}
\end{figure}

1) Let us first consider the anti-diagonal lines and our focus is only on
the first quadrant. The result of the equation~\eqref{eq:x1}, when applied to the new lattice
(where the radius is accordingly $r/\sqrt{2}$ and
$\Delta=1/\sqrt{2}$), implies that there are at least
$r/\sqrt{2}-x_1(r/\sqrt{2},1/\sqrt{2})$
horizontal lines of $\mathbb{Z}^2_e$ (i.e. anti-diagonal lines of the original lattice) which gives at least one vertex to
$\S_{r,1}((0,0),(-1,-1))$. Between these lines, there are the
horizontal lines (anti-diagonal in the original lattice) of the odd
lattice $\mathbb{Z}^2_o$. Hence, the number of consecutive diagonals
in $\mathbb{Z}^2$ contributing a vertex to
 $\S_{r,1}((0,0),(-1,-1))$
 is approaching infinity as $r$ grows.
 Trivially, these contain $s$ consecutive
 anti-diagonal lines intersecting $\S_{r,1}((0,0),(-1,-1))$ in the first
 quadrant, when $r$ is large enough.

 2) Now we consider the diagonal lines (again in the first quadrant). Applying the result on $m(r,\Delta)$ in the equation~\eqref{eq:m} we see  that there is an horizontal line in $\mathbb{Z}^2_e$ (which is an anti-diagonal line in $\mathbb{Z}^2)$ such that the number of its vertices in  $\S_{r,1}((0,0),(-1,-1))$ tends to infinity as $r$ grows. 
 Similarly, there is a vertical line in $\mathbb{Z}^2_o$ with growing number of intersecting vertices. Consequently, there are $s$ consecutive diagonal
 lines intersecting $\S_{r,1}((0,0),(-1,-1))$ in the first quadrant, when $r$ is large enough.

This result implies by symmetry and translation that $\S_{r,\Delta}(u,v)$ contains all the diagonal lines modulo $s$ in the first quadrant (resp. second quadrant) for any $u,v$ with $v=u+(1,1)$ (resp. $v=u+(1,-1)$), $u,v\in \Z^2$.

\medskip

{\bf Claim (ii):} Let us now assume that $\S_{r,\Delta}((-1,0),(0,0))$
contains all the diagonal lines modulo $s$.
Let $C=L^d_s$ (see Proposition~\ref{prop:method}). We show
that $C$ is an $(r,\Delta)$-identifying code. Let
$u=(x_u,y_u)\in\mathbb{Z}^2$ and $v=(x_v,y_v)\in\mathbb{Z}^2$ be two
distinct vertices. Without loss of generality, it suffices to
consider the following two cases:

(a) If $|x_u-x_v|\ge 1$ and $|y_v-y_u|\ge 1$, then by the first claim
there is a diagonal line providing at least one
element of $C$ to $\S_{r,\Delta}(u,v)$.

(b) Let now  $y_v=y_u$  or  $x_v=x_u$, meaning that $u$ and $v$ lie in the same vertical or horizontal line. Then by Lemma~\ref{lem:line}, $\S_{r,\Delta}(u,v)$ is intersecting all the diagonal lines modulo $s$. Hence we are done.
\end{proof}

\medskip

It is in general hard to get good values for $s$ and $t$ in Proposition~\ref{prop:method}. In this section, we give infinite families of values $(r,\Delta)$ for which $s=2,4,6$ or $8$ and $|\S_{r,\Delta}((0,0),(-1,0))|=s$, leading by Proposition~\ref{prop:optimal} to infinite families with optimal codes.

We first start with $s=2$. In this particular case, there is always a code with density $\frac{1}{2}$:

\begin{proposition}\label{prop:code12}
If $r$ is an integer and if $r+\Delta\geq \sqrt{r^2+2r-1}$, then $|\S_{r,\Delta}((0,0),(-1,0))|=2$, and $D(r,\Delta)=\frac{1}{2}$.
\end{proposition}

\begin{proof}
The first part of the proposition is not difficult: if $r$ is an integer, then for any $\Delta<1$, $(-r-1,0)$ and $(r,0)$ are in $\S_{r,\Delta}((0,0),(-1,0))$ and they are the only vertices with ordinate $0$. If $r+\Delta\geq\sqrt{r^2+2r-1}$ and if $(x,y)$ is a vertex of $B_r((0,0))$ with $y\neq 0$, then $x\leq r-1$ and $(x+1)^2+y^2\leq r^2+2r-1 \leq (r+\Delta)^2 $ and so $(x,y)\in B_{r+\Delta}((-1,0))$. Clearly, $\S_{r,\Delta}((0,0),(-1,0))$ is intersecting all the diagonal lines modulo 2.

For the second part of the proposition, we know by Proposition~\ref{prop:optimal} that it will be true for $r$ large enough, but we will next construct a code that is an $(r,\Delta)$-identifying code for any $r$. Figure~\ref{fig:code0.5} gives the construction for $r=4$.

\begin{figure}[h]
\begin{center}
\begin{tikzpicture}[scale=0.5]
\draw[mygrid] (-8.9,-3.9) grid (14.9,3.9);
\foreach \I in {-8,...,14}\foreach \J in {-3,...,3}
	{\node[gridnode](\I\J) at (\I,\J) {};}

\foreach \I in {-8,-6,-3,-1,0,2,4,5,7,10,12}\foreach \J in {-2,0,2}
	\node[code] at (\I,\J) {};
\foreach \I in {-7,-5,-4,-2,1,3,6,8,9,11,13,14}\foreach \J in {-3,-1,1,3}
	\node[code] at (\I,\J) {};	
\draw[ball] (-7.6,-1.6) rectangle  (10.6,0.6);
\node[labelnode] at (1.9,-1) {\color{roug}$X$};
\draw[ball] (-7.4,-1.4) rectangle  (1.4,-0.6);

\end{tikzpicture}
\caption{\label{fig:code0.5} Optimal code of density $\frac{1}{2}$ for $r=4$ when $\sqrt{23}-4\leq \Delta <1$.}
\end{center}
\end{figure}
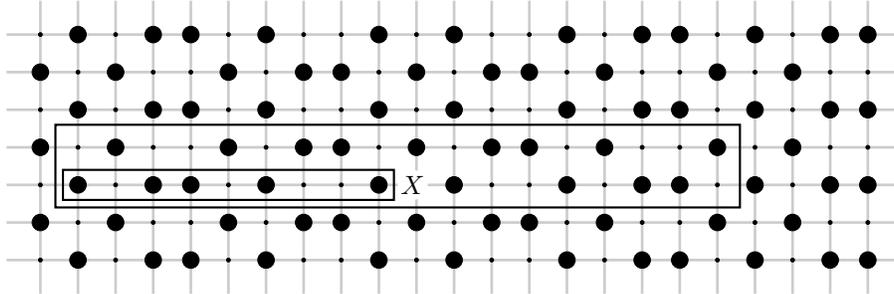

We construct the code with  the following vertical lines of density $1/2$: $L_{0.5}^{(o)}=\{(0,y)\in \Z^2 | y \text{ odd}\}$ and $L_{0.5}^{(e)}=\{(0,y) \in \Z^2 | y \text{ even}\}$. If $r$ is odd, let $X=\{0,1,\ldots,2r\}$. If $r$ is even, let $X=\{0,2,4,\ldots,r-2\}\cup \{r-1\} \cup \{r+1,r+3,\ldots,2r-3\} \cup \{2r\}$  (see Figure~\ref{fig:code0.5}).
Then we construct the subset of $\Z^2$:
\[
U=\{(x,0)| x\equiv i \bmod 4r+2, i \leq 2r \text{ and } x\in X\}\cup \{(x,0) | x\equiv i \bmod 4r+2, i\geq 2r+1 \text{ and } i-(2r+1)\notin X\} \textrm{.}
\]
Finally, we define $\C=\bigcup_{u\in U}(u+L_{0.5}^{(o)}) \cup \bigcup_{u\notin U} (u+L_{0.5}^{(e)})$.

To show that this is an identifying code, we only need to check that it is intersecting the diagonal pattern, for both orientations. Indeed, there is no free $L$-pattern \Lpat in the code (each $L$-pattern contains a vertex of $\C$).  Since $\S_{r,\Delta}(u,v)$  is containing an $L$-pattern whenever $d(u,v)>\sqrt{2}$, it is intersecting $\C$. Furthermore, it is clear that $\C$ is intersecting the horizontal and the vertical pattern, so only the case $d(u,v)=\sqrt{2}$ remains.

We will just prove the result for $v=u+(1,1)$ and $r$ is even, the other cases are similar.
In $\S_{r,\Delta}(u,v)$, there are at least the following four vertices, namely : $A=u+(-r,0)$, $B=u+(0,-r)$, $C=v+(0,r)$ and $D=v+(r,0)$. Because we use the vertical lines $L_{0.5}^{(o)}$ and $L_{0.5}^{(e)}$,
$B$ (resp. $C$) is in the code if and only if $u$ (resp. $v$) is in the code.
Let $(x,y)$ be the coordinates of $A$. If both $A$ and $u$ are not in the code, then either $y$ is even and $x\equiv 2r-1\bmod 4r+2$, or $y$ is odd and  $x\equiv 4r\bmod 4r+2$.
In both cases, it is easy to check that $v$ is in the code, finishing the proof.
\end{proof}

We now consider the case when $\S_{r,\Delta}((0,0),(-1,0))$ has four elements:

\begin{proposition}\label{prop:code38}
Let $(r,\Delta)$ be such that $|\S_{r,\Delta}((0,0),(-1,0))|=4$. There is an $(r,\Delta)$-identifying code of density $\frac{3}{8}$.
\end{proposition}

\begin{proof}
In $\S_{r,\Delta}((0,0),(-1,0))$ there is necessarily a vertex $(a,b)$ with $a\geq b > 0$. Then the other vertices of $\S_{r,\Delta}((0,0),(-1,0))$ are $(a,-b)$,$(-a-1,b)$ and $(-a-1,-b)$.
We write $b=2^kb'$ with $b'$ odd.
We construct the code $C$ in two parts $C_1$ and $C_2$. Define $C_1=\{(x,y)\in \Z^2 | x \text{ and } y \text{ are even}\}$ and $C_2=\{(x,y)\in \Z^2 | (y-x)=i \bmod 2^{k+2} \text{ with } i \in \{2,4,6,\ldots,2^{k+1} \}\}$.
Then the code $C=C_1\cup C_2$ has density $\frac{3}{8}$.

In the diagonal pattern $\S_{r,\Delta}((0,0),(-1,-1))$ there are at least the four vertices $(a,b)$, $(b,a)$, $(a,b-1)$, $(b-1,a)$ and their symmetric images by a central rotation in $(-0.5,-0.5)$. Note that some vertices can be equal (if $a=b$ for example) but in any translation of the pattern, there is always a vertex with both coordinates even. Hence, $C_1$ is intersecting the diagonal pattern. By symmetry, it is also intersecting the anti-diagonal pattern.

We can show that $C_2$ is intersecting the horizontal pattern by noticing that for any $c\in \Z$, among the four values $c+b-a$, $c-b-a$, $c+b+a+1$ and $c-b+a+1$ modulo $2^{k+2}$, there is exactly one of them in $\{2,4,6,\ldots,2^{k+1}\}$.  The same holds by symmetry for the vertical pattern.

Let $u$ and $v$ be a pair of vertices of $\Z^2$ with $d(u,v)>\sqrt{2}$. If $u$ and $v$ are in the same line and $d(u,v)\leq 5$, then by Lemma~\ref{lem:dist}, $\S_{r,\Delta}(u,v)$ is containing a set isomorphic to the horizontal or vertical pattern and so is intersecting $C_2$. Otherwise, $\S_{r,\Delta}(u,v)$ always contains a square of four vertices and so is intersecting $C_1$. Finally $C$ is an $(r,\Delta)$-identifying code.
\end{proof}

In addition, there are some infinite families of $(r,\Delta)$ for which we have optimal codes:

\begin{proposition}
Let $k$ and $i$ be two integers such that $i$ is odd and $i^2<2k+1$. Let $r=\sqrt{k^2+i^2}$ and $r+\Delta=\sqrt{r^2+2k}$.
Then $|\S_{r,\Delta}((0,0),(-1,0))|=4$ and $\S_{r,\Delta}((0,0),(-1,0))$ is intersecting all the diagonal lines modulo 4.
If $r$ is large enough, there is an optimal $(r,\Delta)$-identifying code of density $\frac{1}{4}$.
\end{proposition}

\begin{proof}
To prove that $|\S_{r,\Delta}((0,0),(-1,0))|=4$, we show that $X_1=\S_{r,\Delta}((0,0),(-1,0))\cap\{(x,y)|x>0,y\geq 0\}$ has only one vertex, and that this vertex has strictly positive ordinate.
Let $(x,y) \in X_1$, we have $x\leq \lfloor r \rfloor = k$. If $x=k$ then $y\leq i$. If $y<i$ then $(x+1)^2+y^2\leq (k+1)^2+i^2-1\leq r^2+2k = (r+\Delta)^2$, so $(x,y)\notin X_1$. If $y=i$ then $(x+1)^2+y^2=r^2+2k+1>(r+\Delta)^2$ and so $(k,i)\in X_1$.
If $x<k$, then $(x+1)^2+y^2\leq r^2+2k-1\leq (r+\Delta)^2$ and again $(x,y)\notin X_1$. Therefore, the only vertex in $X_1$ is $(k,i)$, and $i>0$. This implies that $|\S_{r,\Delta}((0,0),(-1,0))|=4$
and the four vertices of $\S_{r,\Delta}((0,0),(-1,0))$ are $(k,i)$, $(k,-i)$, $(-k-1,i)$ and $(-k-1,-i)$.
It remains to show that the four values $i-k$, $-i-k$, $i+k+1$, $-i+k+1$ are different modulo 4. By adding $i+k$, it is the same to show that the four values $0$, $2i$, $2k+1$, $2k+1+2i$ are different modulo 4, which is clear. The last claim is a direct consequence of Proposition~\ref{prop:optimal}.
\end{proof}

\begin{proposition}
Let $k$ be an odd integer not divisible by 3 and let $r=2k^2+1$, $r+\Delta=\sqrt{r^2+2r-3}$.
Then $|\S_{r,\Delta}((0,0),(-1,0))|=6$ and $\S_{r,\Delta}((0,0),(-1,0))$ is intersecting all the diagonal lines modulo $6$.
If $r$ is large enough, there is an optimal $(r,\Delta)$-identifying code of density $\frac{1}{6}$.
\end{proposition}

\begin{proof}
Let again  $X_1=\S_{r,\Delta}((0,0),(-1,0))\cap\{(x,y)|x>0,y\geq 0\}$. The only vertices of $X_1$ are $(r,0)$ and $(r-1,2k)$. This implies that  $|\S_{r,\Delta}((0,0),(-1,0))|=6$ and that the $6$ vertices of $\S_{r,\Delta}((0,0),(-1,0))$ are $(r,0)$,$(-r-1,0)$, $(r-1,2k)$,$(r-1,-2k)$,$(-r,2k)$,$(-r,-2k)$.
Then the six values $-r,r+1,2k-r+1,-2k-r+1,2k+r,-2k+r$ are all different modulo $6$, and we can conclude with Proposition~\ref{prop:optimal}.
\end{proof}

\begin{proposition}
Let $k\geq 18$ be an integer such that $k\equiv2\bmod 16$ and let $L=\left(\frac{k}{2}\right)^2-1$.
Let $r=\sqrt{L^2+8}$ and $r+\Delta=\sqrt{L^2+2L+4}$.
Then $|\S_{r,\Delta}((0,0),(-1,0))|=8$ and $\S_{r,\Delta}((0,0),(-1,0))$ is intersecting all the diagonal lines modulo $8$.
If $r$ is large enough, there is an optimal $(r,\Delta)$-identifying code of density $\frac{1}{8}$.
\end{proposition}

\begin{proof}
We first show that the vertices of $X_1=\S_{r,\Delta}((0,0),(-1,0))\cap\{(x,y)|x>0,y\geq 0\}$ are $(L,2)$ and $(L-2,k)$. Those two vertices clearly are in $X_1$.
It is also clear that there are no other vertices with abscissa $L$ or at most $L-2$.
It remains to show that there are no vertices with abscissa $L-1$. If there will be a vertex $(L-1,y)$ in $X_1$ then necessarily, $y^2\leq 2L+7$. But $2L\equiv 0 \bmod 16$, so $2L+5$, $2L+6$ and $2L+7$ are not squares of integers, and so $y^2\leq 2L+4$. Then $(x+1)^2+y^2\leq (r+\Delta)^2$, a contradiction. A simple computation shows that the eight vertices of $\S_{r,\Delta}((0,0),(-1,0))$ are on different diagonal lines modulo $8$, and again we conclude with Proposition~\ref{prop:optimal}.
\end{proof}

\section{Study of small values of $(r,\Delta)$} \label{SectionSmall}

In the previous sections, we have considered $(r,\Delta)$-identifying codes with general (large) values of $r$. In this section, we focus on some specific small values of $r$. As one can expect, fixing the values $r$ and $\Delta$ leads to more efficient constructions of $(r, \Delta)$-identifying codes as well as better lower bounds. In particular, it should be mentioned that the methods for obtaining lower bounds in these special cases drastically differ from the ones used in the general cases. The results of this section as well as some previously known results have been gathered in Table~\ref{table:small}.

\begin{table}[ht]
  \begin{displaymath}
\begin{array}{|c|c|c|c|c|c|c|c|c|}
\hline
   r \backslash r+\Delta & 1&\sqrt{2}&2& \sqrt{5}&\sqrt{8} &3&\sqrt{10} \\ \hline
    1 & 0.35&0.5^{b,c}&X& X&X&X&X\\ \hline
     \sqrt{2}& - &\frac{2}{9}^a&[\frac{16}{57},\frac{1}{3}]& X&X&X&X\\ \hline
      2 & - & - &[0.15,0.17]^a&0.5^{b,c}&0.5^{b,c}&X&X\\ \hline
       \sqrt{5} & - & - & -&0.125^a&[0.17,\frac{2}{9}]&[0.25^b,\frac{1}{3}^d]&X\\ \hline
        \sqrt{8} & - & -& - & -&0.125^a&[\frac{1}{7},\frac{4}{21}]&[0.25^b,0.375^d]\\ \hline
\end{array}
\end{displaymath}
\begin{multicols}{2}
X: No $(r,\Delta)$-identifying code \\
a: See \cite{JL11} and the references therein\\
b: Trivial lower bound of Proposition~\ref{prop:lowerbound} \\
c: Code of Proposition~\ref{prop:code12} \\
d: Code from Subsection~\ref{sec:goodcode}\\
\end{multicols}
\caption{Bounds for small values of $r$.}
\label{table:small}
\end{table}

\subsection{Case $(r,\Delta)=(\sqrt{2},2-\sqrt{2})$}

In Figures~\ref{fig:hor2} and~\ref{fig:diag2}, the horizontal and diagonal patterns for $(r,\Delta)=(\sqrt{2},2-\sqrt{2})$ are shown. That directly gives, with Proposition~\ref{prop:lowerbound}, the lower bound $D(\sqrt{2},2-\sqrt{2})\geq \frac{1}{4}$. We can improve this lower bound using discharging methods. We first sketch a proof of an easy improvement on this bound to give an idea of the used method.

\begin{proposition}\label{prop:low2f}
We have
\[D(\sqrt{2},2-\sqrt{2})\geq \frac{4}{15} \textrm{.}\]
\end{proposition}

\begin{figure}[h]
\begin{center}
\subfloat[][\label{fig:hor2}]{
\begin{tikzpicture}[scale=0.6]
\draw[mygrid] (-2.9,-1.9) grid (1.9,1.9);
\foreach \I in {-2,...,1}\foreach \J in {-1,...,1}
	{\node[gridnode](\I\J) at (\I,\J) {};}
\foreach \pos in {(-2, 1), (-2,-1), (1, 1),(1,-1)}
	\node[pattern] at \pos {};
	
\node[centernode] at (0,0) {};
\node[labelnode] at (0.3,0.5) {$v$};
\node[centernode] at (-1,0) {};
\node[labelnode] at (-1.3,0.5) {$u$};
\end{tikzpicture}
}
\hfil
\subfloat[][\label{fig:diag2}]{
\begin{tikzpicture}[scale=0.6]
\draw[mygrid] (-2.9,-2.9) grid (1.9,1.9);
\foreach \I in {-2,...,1}\foreach \J in {-2,...,1}
	{\node[gridnode](\I\J) at (\I,\J) {};}

\draw[patternpath] (0.5,0.5)--(-0.5,0.5)--(-0.5,1.5)--(1.5,1.5)--(1.5,-0.5)--(0.5,-0.5)--(0.5,0.5);
	
\draw[patternpath, rotate = 180, shift = {(1,1)}] (0.5,0.5)--(-0.5,0.5)--(-0.5,1.5)--(1.5,1.5)--(1.5,-0.5)--(0.5,-0.5)--(0.5,0.5);
\node[centernode] at (0,0) {};
\node[labelnode] at (0,-0.5) {$v$};
\node[centernode] at (-1,-1) {};
\node[labelnode] at (-1,-0.5) {$u$};
\end{tikzpicture}}
\hfil
\subfloat[][\label{fig:frame3}]{
\begin{tikzpicture}[scale=0.6]
\draw[mygrid] (-2.9,-2.9) grid (1.9,1.9);
\foreach \I in {-2,...,1}\foreach \J in {-2,...,1}
	{\node[gridnode](\I\J) at (\I,\J) {};}

\begin{scope}[even odd rule]
\filldraw[pattern] (-2.5,-2.5) rectangle (1.5,1.5)
 (-1.5,-1.5) rectangle (0.5,0.5);
\end{scope}
\end{tikzpicture}}
\end{center}
\caption{The set $\S_{\sqrt{2}, 2-\sqrt{2}}(u,v)$ when (a) $v-u=(1,0)$  and (b) $v-u=(1,1)$.  In (c), the frame used to increase the lower bound in the case $(r,\Delta)=(\sqrt{2},2-\sqrt{2})$.}
\end{figure}
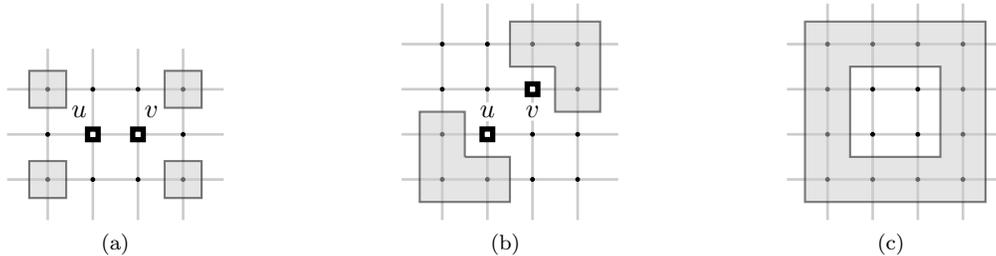

\begin{proof}
To prove this lower bound, we use the frame of Figure~\ref{fig:frame3}. Let $F$ be a fixed set of vertices of $\Z^2$ forming the frame of Figure~\ref{fig:frame3}.
Let $C$ be a $(\sqrt{2},2-\sqrt{2})$-identifying code. In what follows, we show that on average there are at least $16/5$ vertices of the code in every frame. We first note that for any $u\in \Z^2$, $C\cap (u+F)$ contains at least three vertices. Indeed, $F$ contains two disjoint horizontal patterns, so $C\cap (u+F)$ contains at least two vertices. Assume there are only two vertices, then one of the vertices must be in a corner, otherwise one horizontal or vertical pattern is empty. Then the other one must be in the opposite corner, but that implies that a diagonal pattern is empty.

\begin{figure}[h]
\begin{center}
 \subfloat[][\label{fig:F3a}]
{\begin{tikzpicture}[scale=0.45]
\draw[mygrid] (-2.9,-2.9) grid (1.9,1.9);
\foreach \I in {-2,...,1}\foreach \J in {-2,...,1}
	{\node[gridnode](\I\J) at (\I,\J) {};}

\begin{scope}[even odd rule]
\filldraw[pattern] (-2.5,-2.5) rectangle (1.5,1.5)
 (-1.5,-1.5) rectangle (0.5,0.5);
\end{scope}
\foreach \pos in {(-2,1),(-1,-2),(1,-2)}
	{\node[code] at \pos {};}

\end{tikzpicture}}
\hfil
 \subfloat[][\label{fig:F3b}]{
\begin{tikzpicture}[scale=0.45]
\draw[mygrid] (-2.9,-2.9) grid (1.9,1.9);
\foreach \I in {-2,...,1}\foreach \J in {-2,...,1}
	{\node[gridnode](\I\J) at (\I,\J) {};}

\begin{scope}[even odd rule]
\filldraw[pattern] (-2.5,-2.5) rectangle (1.5,1.5)
 (-1.5,-1.5) rectangle (0.5,0.5);
\end{scope}

\foreach \pos in {(-2,1),(-2,-2),(1,1)}
	{\node[code] at \pos {};}
\end{tikzpicture}
}
\hfil
 \subfloat[][\label{fig:F3c}]{
\begin{tikzpicture}[scale=0.45]
\draw[mygrid] (-2.9,-2.9) grid (1.9,1.9);
\foreach \I in {-2,...,1}\foreach \J in {-2,...,1}
	{\node[gridnode](\I\J) at (\I,\J) {};}

\begin{scope}[even odd rule]
\filldraw[pattern] (-2.5,-2.5) rectangle (1.5,1.5)
 (-1.5,-1.5) rectangle (0.5,0.5);
\end{scope}
\foreach \pos in {(-2,1),(-2,0),(1,1)}
	{\node[code] at \pos {};}
\end{tikzpicture}}
\hfil
 \subfloat[][\label{fig:F3d}]{
\begin{tikzpicture}[scale=0.45]
\draw[mygrid] (-2.9,-2.9) grid (1.9,1.9);
\foreach \I in {-2,...,1}\foreach \J in {-2,...,1}
	{\node[gridnode](\I\J) at (\I,\J) {};}

\begin{scope}[even odd rule]
\filldraw[pattern] (-2.5,-2.5) rectangle (1.5,1.5)
 (-1.5,-1.5) rectangle (0.5,0.5);
\end{scope}
\foreach \pos in {(-2,1),(-1,1),(1,0)}
	{\node[code] at \pos {};}
\end{tikzpicture}}
\hfil
 \subfloat[][\label{fig:F3e}]{
\begin{tikzpicture}[scale=0.45]
\draw[mygrid] (-2.9,-2.9) grid (1.9,1.9);
\foreach \I in {-2,...,1}\foreach \J in {-2,...,1}
	{\node[gridnode](\I\J) at (\I,\J) {};}

\begin{scope}[even odd rule]
\filldraw[pattern] (-2.5,-2.5) rectangle (1.5,1.5)
 (-1.5,-1.5) rectangle (0.5,0.5);
\end{scope}
\foreach \pos in {(-2,1),(-1,-2),(1,0)}
	{\node[code] at \pos {};}
\end{tikzpicture}
}
\end{center}
\caption{\label{fig:possframe3} In a $(\sqrt{2},2-\sqrt{2})$-identifying code, there are only those five possibilities for the frame of Figure~\ref{fig:frame3} to have three vertices. For each of them, there is a frame in the neighbourhood with at least four vertices.}
\end{figure}
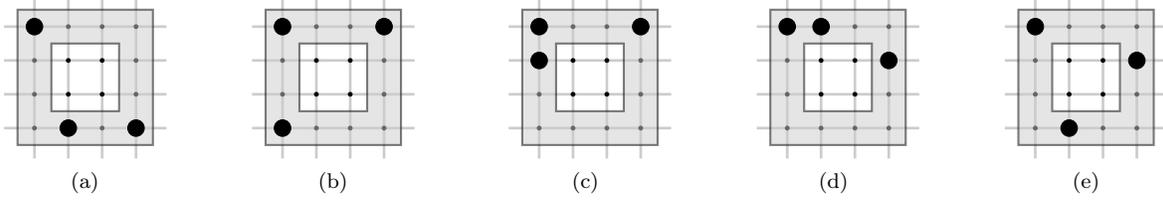

In fact, if $C\cap (u+F)$ contains exactly three vertices, there are, up to obvious symmetry, only five different possibilities for the positions of the three vertices, depicted in Figure~\ref{fig:possframe3}. We can observe that for any of those possibilities, one of the neighbouring frames, i.e. one set $v+F$ with $d(u,v)=1$, is containing at least four vertices of $C$: the frame on the top for case (c) and the frame on the left for the other cases.

We now give, for any $u\in \Z^2$, charge $|C\cap (u+F)|$ to each set $u+F$. We apply the following rule in order to even out the charge among the frames: each set $u+F$ with charge at least 4 gives charge $\frac{1}{5}$ away to the neighbouring sets $v+F$, with $d(u,v)=1$ which have charge 3. By the previous remark, after this process, each set $u+F$ will have charge at least $3+\frac{1}{5}=\frac{16}{5}$.

That means that on average, the number of vertices of $C$ in a frame $u+F$ is at least $\frac{16}{5}$. Using the same method as in the proof of Proposition~\ref{prop:lowpattern} (see \cite{H04} and \cite{JL11}), we can conclude that the density of $C$ is at least $\frac{16}{5}\times \frac{1}{|F|} = \frac{4}{15}$.
\end{proof}

We can improve this lower bound by further analysis and more advanced discharging rules. That leads to the following proposition, which is shown in Appendix B:

\begin{proposition}\label{prop:low2}
We have
\[D(\sqrt{2},2-\sqrt{2})\geq \frac{16}{57}\textrm{.}\]
\end{proposition}

\medskip

For the upper bound, we use the code $C$ of Figure~\ref{fig:code2} that has density $\frac{1}{3}$.
To show that $C$ is a $(\sqrt{2},2-\sqrt{2})$-identifying code, we only need, in this particular case, to check that it is a $\sqrt{2}$-dominating set and that the sets $\S_{\sqrt{2},2-\sqrt{2}}(u,v)$ are intersecting $C$ for $d(u,v)\leq \sqrt{2}$. Indeed, for the other pairs of vertices, $\S_{\sqrt{2},2-\sqrt{2}}(u,v)$ is either containing  a set $\S_{\sqrt{2},2-\sqrt{2}}(u',v')$ with $d(u',v')\leq \sqrt{2}$ or the set $B_r(u)$. This is not true for general values of $r$ and $\Delta$.

\begin{figure}[h]
\begin{center}
\begin{tikzpicture}[scale=0.6,rotate = 90]
\draw[mygrid] (-4.9,-5.9) grid (4.9,5.9);
\foreach \I in {-4,...,4}\foreach \J in {-5,...,5}
	{\node[gridnode](\I\J) at (\I,\J) {};}

\foreach \pos in {(-4,-5),(0,-5),(-2,-4),(-1,-4),(-3,-3),(1,-3)}
	{\node[code] at \pos {};}
	
\begin{scope}[shift={(3,3)}]
\foreach \pos in {(-4,-5),(0,-5),(-2,-4),(-1,-4),(-3,-3),(1,-3)}
	{\node[code] at \pos {};}
\end{scope}

\begin{scope}[shift={(0,6)}]
\foreach \pos in {(-4,-5),(0,-5),(-2,-4),(-1,-4),(-3,-3),(1,-3)}
	{\node[code] at \pos {};}
\end{scope}

\begin{scope}[shift={(3,9)}]
\foreach \pos in {(-4,-5),(0,-5),(-2,-4),(-1,-4)}
	{\node[code] at \pos {};}
\end{scope}

\begin{scope}[shift={(6,0)}]
\foreach \pos in {(-4,-5),(-2,-4),(-3,-3)}
	{\node[code] at \pos {};}
\end{scope}

\begin{scope}[shift={(-3,3)}]
\foreach \pos in {(0,-5),(-1,-4),(1,-3)}
	{\node[code] at \pos {};}
\end{scope}

\begin{scope}[shift={(6,6)}]
\foreach \pos in {(-4,-5),(-2,-4),(-3,-3)}
	{\node[code] at \pos {};}
\end{scope}

\begin{scope}[shift={(-3,9)}]
\foreach \pos in {(0,-5),(-1,-4)}
	{\node[code] at \pos {};}
\end{scope}

\draw[ball] (-1.5,-1.5) rectangle (1.5,4.5);
\draw[ball,dashed] (-4.5,-4.5) rectangle (-1.5,1.5);
\end{tikzpicture}
\end{center}
\caption{\label{fig:code2}A $(\sqrt{2},2-\sqrt{2})$-identifying code of density $1/3$.}
\end{figure}
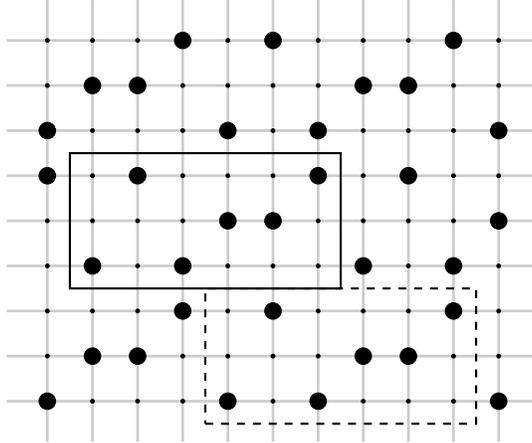

\subsection{Case $(r,\Delta)=(\sqrt{5},\sqrt{8}-\sqrt{5})$}

In Figures~\ref{fig:hor5} and~\ref{fig:diag5}, the horizontal and diagonal patterns for $(r,\Delta)=(\sqrt{5},\sqrt{8}-\sqrt{5})$ are shown. As before, we can improve the straightforward lower bound of $\frac{1}{6}$ using the frame of Figure~\ref{fig:low5} and discharging rules:

\begin{proposition}\label{prop:low5}
We have
\[D(\sqrt{5},\sqrt{8}-\sqrt{5})\geq 0.17. \]
\end{proposition}

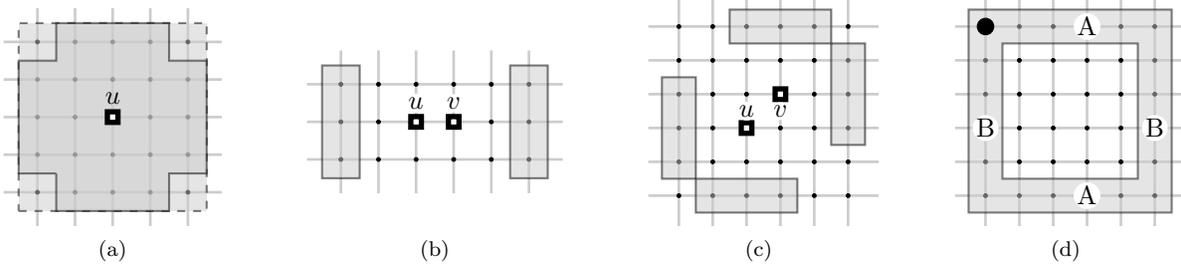
\begin{figure}[ht]
\begin{center}
\subfloat[][\label{fig:b5}]{
\begin{tikzpicture}[scale=0.5]
\draw[mygrid] (-2.9,-2.9) grid (2.9,2.9);
\foreach \I in {-2,...,2}\foreach \J in {-2,...,2}
	{\node[gridnode](\I\J) at (\I,\J) {};}
\draw[patternpath, dashed] (-2.5,-2.5) rectangle (2.5,2.5);
\draw[patternpath] (-1.5,-1.5) -- (-2.5,-1.5) -- (-2.5,1.5) -- (-1.5,1.5) -- (-1.5,2.5)--(1.5,2.5) -- (1.5,1.5) -- (2.5,1.5) -- (2.5,-1.5) -- (1.5,-1.5) -- (1.5,-2.5) -- (-1.5,-2.5) -- (-1.5,-1.5);
\node[centernode] at (0,0) {};
\node[] at (0,0.5) {$u$};
\end{tikzpicture}
}
\hfil
\subfloat[][\label{fig:hor5}]{
\begin{tikzpicture}[scale=0.5, baseline = -40]
\draw[mygrid] (-3.9,-1.9) grid (2.9,1.9);
\foreach \I in {-3,...,2}\foreach \J in {-1,...,1}
	{\node[gridnode](\I\J) at (\I,\J) {};}

\draw[patternpath] (-2.5,-1.5) rectangle (-3.5,1.5);
\draw[patternpath] (2.5,-1.5) rectangle (1.5,1.5);
	
\node[centernode] at (0,0) {};
\node[circle,fill=white, inner sep =0] at (0,0.5) {$v$};
\node[centernode] at (-1,0) {};
\node[circle,fill=white, inner sep =0] at (-1,0.5) {$u$};
\end{tikzpicture}
}
\hfil
\subfloat[][\label{fig:diag5}]{
\begin{tikzpicture}[scale=0.45]
\draw[mygrid] (-3.9,-3.9) grid (2.9,2.9);
\foreach \I in {-3,...,2}\foreach \J in {-3,...,2}
	{\node[gridnode](\I\J) at (\I,\J) {};}

\draw[patternpath] (-3.5,-2.5) rectangle (-2.5,0.5);
\draw[patternpath] (2.5,-1.5) rectangle (1.5,1.5);
\draw[patternpath] (-2.5,-2.5) rectangle (0.5,-3.5);
\draw[patternpath] (-1.5,2.5) rectangle (1.5,1.5);
\node[centernode] at (0,0) {};
\node[labelnode] at (0,-0.5) {$v$};
\node[centernode] at (-1,-1) {};
\node[labelnode] at (-1,-0.5) {$u$};
\end{tikzpicture}}
\hfil
\subfloat[][\label{fig:low5}]{
\begin{tikzpicture}[scale=0.45]
\draw[mygrid] (-3.9,-3.9) grid (2.9,2.9);
\foreach \I in {-3,...,2}\foreach \J in {-3,...,2}
	{\node[gridnode](\I\J) at (\I,\J) {};}

\begin{scope}[even odd rule]
\filldraw[pattern] (-3.5,-3.5) rectangle (2.5,2.5)
 (-2.5,-2.5) rectangle (1.5,1.5);
\end{scope}
\node[labelnode] at (0,2) {A};
\node[labelnode] at (0,-3) {A};
\node[labelnode] at (2,-1) {B};
\node[labelnode] at (-3,-1) {B};
\node[code] at (-3,2) {};
\end{tikzpicture}}
\end{center}
\caption{In (a), the set $B_{r}(u)$ for $r=\sqrt{5}$ and $r=\sqrt{8}$. In (b) and (c), the set $\S_{\sqrt{5}, \sqrt{8}-\sqrt{5}}(u,v)$ when (b) $v-u=(1,0)$ and (c) $v-u=(1,1)$.  In (d), the frame used to increase the lower bound in the case $(r,\Delta)=( \sqrt{5},\sqrt{8}-\sqrt{5})$.
}
\end{figure}

\begin{proof}
We use frame $F$ of Figure~\ref{fig:low5}. The proof is based on showing that on average each such frame contains at least $17/5$ vertices of a code. First one can notice that in a $(\sqrt{5},\sqrt{8}-\sqrt{5})$-identifying code, there are at least three vertices in each translation of frame $F$. Furthermore, if there are only three vertices in one translation $u+F$ of $F$, then, necessarily, one of the three vertices, say $c$, is in a corner and without loss of generality, we can assume it is in the top left corner. Then there must be one vertex of the code in position $c+(3,0)$ or $c+(3,-6)$ (positions $A$ in Figure~\ref{fig:low5}) and one vertex in position $c+(0,-3)$ or $c+(6,-3)$ (positions $B$ in the figure). Then the frame on the left and on the top have necessarily four vertices.
This is not enough to improve the lower bound of $\frac{1}{6}$, but we can use further analysis.
Let $C$ be a $(\sqrt{5},\sqrt{8}-\sqrt{5})$-identifying code.
Let $\mathcal F_{3A}$ (resp. $\mathcal F_{3B}$) be all the sets $u+F$ such  that $|(u+F)\cap C|=3$ and there are exactly two (resp. at least three) neighbouring frames $v+F$ with $d(u,v)=1$ such that $|(v+F)\cap C|\geq 4$. Let $\mathcal F_{i}$ (resp. $\mathcal F_{\geq i}$) be all the sets $u+F$ such  that $|(u+F)\cap C|=i$ (resp. $|(u+F)\cap C|\geq i$).
We have the following facts:
\begin{enumerate}
\item If a set $u+F\in \mathcal F_{\geq 4}$ has four neighbours in $\mathcal F_{3A}$, then it has at least five elements of $C$.
\item A set $u+F\in \mathcal F_{4}$ cannot have three neighbours in $\mathcal F_{3A}$ and one in  $\mathcal F_{3B}$.
\item If a set $u+F\in \mathcal F_{3B}$, then either it has a neighbour in $\mathcal F_{\geq 4}$ which has a neighbour in $\mathcal F_{\geq 4}$, or it has four neighbours in $\mathcal F_{\geq 4}$.
\end{enumerate}

Indeed, let us for example show the first fact. Let us assume there exists $u+F\in \mathcal F_{4}$ with four neighbours in $\mathcal F_{3A}$. Among the four frames $v+F$ with $d(u,v)=\sqrt{2}$ (frames in diagonal), exactly two of them are in $\mathcal F_{\geq 4}$, and they are in diagonal. Without loss of generality, we can assume that $(u+(1,1))+F$ and $(u+(-1,-1))+F$ are in $\mathcal F_{\geq 4}$. Then the four vertices of the code in $u+F$ are fixed. Indeed, each neighbouring frame of $u+F$ has its corner fixed, and it is in $u+F$ (the corner is fixed by the position of the two neighbouring frames in  $\mathcal F_{\geq 4}$). But then a diagonal pattern in $u + F$ is free (it contains no vertex of the code) in $u+F$, a contradiction. The other facts are proved in a similar way.

Now we give charge $|(u+F)\cap C|$ to each set $u+F$. Each set $u+F$ of $\mathcal F_{\geq 4}$ gives charge $\frac{1}{5}$ to each neighbour of $\mathcal F_{3A}$ and gives to each neighbour of $\mathcal F_{3B}$:
\begin{itemize}
\item $\frac{1}{10}$ if $u+F$ has no neighbours in $\mathcal F_{\geq 4}$,
\item $\frac{1}{5}$ if $u+F$ has one neighbour in $\mathcal F_{\geq 4}$.
\end{itemize}

Then a set of $\mathcal F_{\geq 4}$ gives at most $\frac{3}{5}$ if it has four elements of the code, and $\frac{4}{5}$ otherwise, and each set of $\mathcal F_{3A}$ or $\mathcal F_{3B}$ receives at least $\frac{2}{5}$.
After the discharging, each set $u+F$ has at least charge $3+\frac{2}{5}=\frac{17}{5}$. That means that in the pattern $F$, there are in average at least $\frac{17}{5}$ vertices of the code. As before, there are $20$ vertices in $F$ so $D(C)\geq \frac{17}{100}=0.17$.
\end{proof}

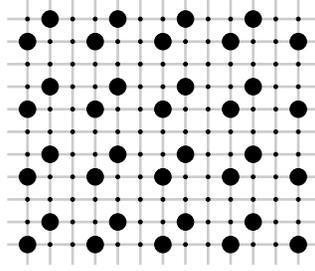
\begin{figure}[h]
\begin{center}
\begin{tikzpicture}[scale=0.3]
\draw[mygrid] (-4.9,-5.9) grid (8.9,5.9);
\foreach \I in {-4,...,8}\foreach \J in {-5,...,5}
	{\node[gridnode](\I\J) at (\I,\J) {};}

\foreach \I in {-4,-1,2,5,8} \foreach \J in {-5,-2,1,4}
	\node[code] at (\I,\J) {};
\foreach \I in {-3,0,3,6} \foreach \J in {-4,-1,2,5}
	\node[code] at (\I,\J) {};
	
\end{tikzpicture}
\end{center}
\caption{\label{fig:code5}A $(\sqrt{5},\sqrt{8}-\sqrt{5})$-identifying code of density $\frac{2}{9}$.}
\end{figure}

We believe that the lower bound in the previous proposition is not the optimal one. We can certainly improve it  with further analysis of this frame, but we think that will lead  to very small improvements.
For upper bound, we use the code of Figure~\ref{fig:code5} of density $\frac{2}{9}$.
Here, the code is clearly a $\sqrt{5}$-dominating set, and even a $2$-dominating set. It is also intersecting the horizontal and diagonal patterns (with rotations). If $u$ and $v$ are distinct vertices of $\Z^2$, we can prove that either $\S_{\sqrt{5},\sqrt{8}-\sqrt{5}}(u,v)$ contains an horizontal or a diagonal pattern, or it contains a ball of radius 2, showing that $C$ is a $(\sqrt{5},\sqrt{8}-\sqrt{5})$-identifying code.

\subsection{Case $(r,\Delta)=(\sqrt{8},3-\sqrt{8})$}

In Figures~\ref{fig:hor8} and~\ref{fig:diag8}, the horizontal and diagonal patterns for $(r,\Delta)=(\sqrt{8},3-\sqrt{8})$ are shown. Again, we can improve the straightforward lower bound of $\frac{1}{8}$ using the frame of Figure~\ref{fig:low8}:

\begin{proposition}\label{prop:low8}
We have
\[D(\sqrt{8},3-\sqrt{8})\geq \frac{1}{7}\textrm{.}\]
\end{proposition}

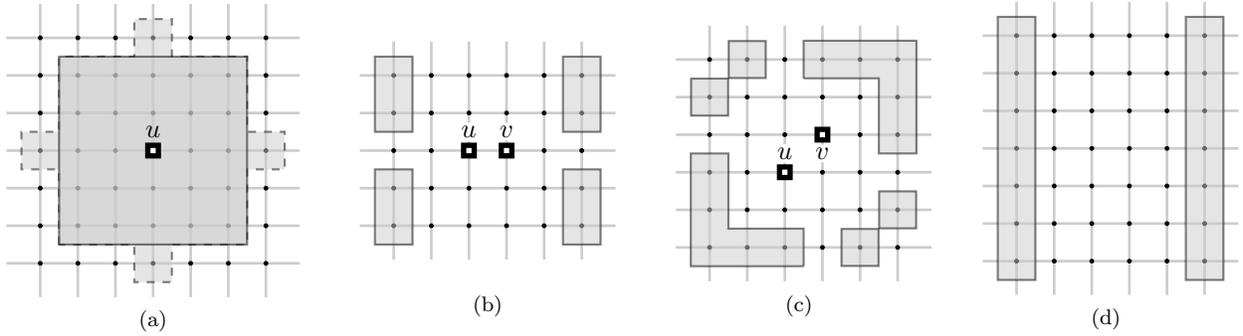
\begin{figure}[h]
\begin{center}
\subfloat[][\label{fig:b8}]{
\begin{tikzpicture}[scale=0.5,baseline=-50]
\draw[mygrid] (-3.9,-3.9) grid (3.9,3.9);
\foreach \I in {-3,...,3}\foreach \J in {-3,...,3}
	{\node[gridnode](\I\J) at (\I,\J) {};}

\draw[patternpath,dashed] (-2.5,-2.5) --(-2.5,-0.5)--(-3.5,-0.5)--(-3.5,0.5)--(-2.5,0.5) -- (-2.5,2.5) -- (-0.5,2.5)--(-0.5,3.5)--(0.5,3.5)--(0.5,2.5)-- (2.5,2.5) -- (2.5,0.5) --(3.5,0.5)--(3.5,-0.5)--(2.5,-0.5) -- (2.5,-2.5) -- (0.5,-2.5)--(0.5,-3.5)--(-0.5,-3.5)--(-0.5,-2.5)--(-2.5,-2.5);
\draw[patternpath] (-2.5,-2.5) rectangle (2.5,2.5);
\node[centernode] at (0,0) {};
\node[] at (0,0.5) {$u$};
\end{tikzpicture}
}
\hfil
\subfloat[][\label{fig:hor8}]{
\begin{tikzpicture}[scale=0.5,baseline =-50]
\draw[mygrid] (-3.9,-2.9) grid (2.9,2.9);
\foreach \I in {-3,...,2}\foreach \J in {-2,...,2}
	{\node[gridnode](\I\J) at (\I,\J) {};}

\draw[patternpath] (-3.5,-2.5) rectangle (-2.5,-0.5);
\draw[patternpath] (-3.5,0.5) rectangle (-2.5,2.5);
\draw[patternpath] (1.5,-2.5) rectangle (2.5,-0.5);
\draw[patternpath] (1.5,0.5) rectangle (2.5,2.5);
	
\node[centernode] at (0,0) {};
\node[circle,fill=white, inner sep =0] at (0,0.5) {$v$};
\node[centernode] at (-1,0) {};
\node[circle,fill=white, inner sep =0] at (-1,0.5) {$u$};
\end{tikzpicture}
}
\hfil
\subfloat[][\label{fig:diag8}]{
\begin{tikzpicture}[scale=0.5]
\draw[mygrid] (-3.9,-3.9) grid (2.9,2.9);
\foreach \I in {-3,...,2}\foreach \J in {-3,...,2}
	{\node[gridnode](\I\J) at (\I,\J) {};}

\draw[patternpath] (-0.5,1.5)--(-0.5,2.5)--(2.5,2.5)--(2.5,-0.5)--(1.5,-0.5)--(1.5,1.5)--(-0.5,1.5);
\draw[patternpath] (-0.5,-3.5)--(-0.5,-2.5)--(-2.5,-2.5)--(-2.5,-0.5)--(-3.5,-0.5)--(-3.5,-3.5)--(-0.5,-3.5);
\node[pattern] at (-2,2) {};
\node[pattern] at (-3,1) {};
\node[pattern] at (2,-2) {};
\node[pattern] at (1,-3) {};

\node[centernode] at (0,0) {};
\node[labelnode] at (0,-0.5) {$v$};
\node[centernode] at (-1,-1) {};
\node[labelnode] at (-1,-0.5) {$u$};
\end{tikzpicture}}
\hfil
\subfloat[][\label{fig:low8}]{
\begin{tikzpicture}[scale=0.5,baseline=-65]
\draw[mygrid] (-3.9,-4.9) grid (2.9,2.9);
\foreach \I in {-3,...,2}\foreach \J in {-4,...,2}
	{\node[gridnode](\I\J) at (\I,\J) {};}
\draw[pattern] (-3.5,-4.5) rectangle (-2.5,2.5);
\draw[pattern] (1.5,-4.5) rectangle (2.5,2.5);
\end{tikzpicture}}
\end{center}
\caption{In (a), the set $B_{r}(u)$ for $r=\sqrt{8}$ and $r=3$. In (b) and (c), the set $\S_{\sqrt{8}, 3-\sqrt{8}}(u,v)$ when (b) $v-u=(1,0)$ and (c) $v-u=(1,1)$.  In (d), the frame used to increase the lower bound in the case $(r,\Delta)=( \sqrt{8},3-\sqrt{8})$.}
\end{figure}

\begin{proof}
Let $F$ be the frame of Figure~\ref{fig:low8} and let $C$ be a  $(\sqrt{8},3-\sqrt{8})$-identifying code. Then for any vertex $u\in \Z^2$, $(u+F)\cap C$ must have at least two vertices. And so the density of $C$ is at least $\frac{1}{7}$.
\end{proof}

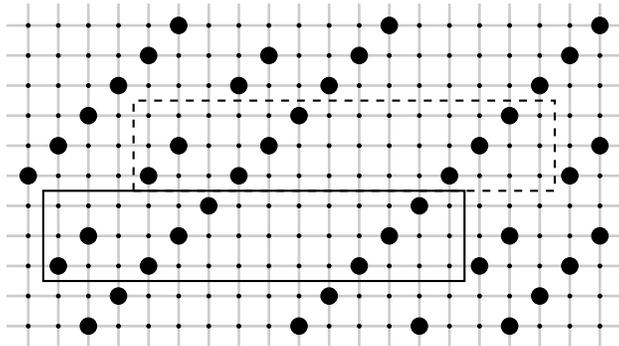
\begin{figure}[ht]
\begin{center}
\begin{tikzpicture}[scale=0.4]
\clip (-7.7,-5.7) rectangle (12.7,5.7);
\draw[mygrid] (-7.9,-5.9) grid (12.9,5.9);
\foreach \I in {-7,...,12}\foreach \J in {-5,...,5}
	{\node[gridnode](\I\J) at (\I,\J) {};}

\foreach \pos in {(-7, 0), (-6, -6), (-6, -3), (-6, 1), (-5, -5), (-5, -2), (-5, 2), (-4, -4), (-4, 3), (-3, -3), (-3, 0), (-3, 4), (-2, -2), (-2, 1), (-2, 5), (-1, -1), (0, 0), (0, 3), (1, -6), (1, 1), (1, 4), (2, -5), (2, 2), (3, -4), (3, 3), (4, -3), (4, 4), (5, -6), (5, -2), (5, 5), (6, -5), (6, -1), (7, 0), (8, -6), (8, -3), (8, 1), (9, -5), (9, -2), (9, 2), (10, -4), (10, 3), (11, -3), (11, 0), (11, 4), (12, -2), (12, 1), (12, 5), (13, -1)}
	\node[code] at \pos {};
	
\draw[ball] (-6.5,-3.5) rectangle (7.5,-0.5);
\draw[ball, dashed] (-3.5,-0.5) rectangle (10.5,2.5);
	
\end{tikzpicture}
\end{center}
\caption{Code of density $4/21$ for $(r,\Delta)=(\sqrt{8},3-\sqrt{8})$.}
\label{fig:code8}
\end{figure}

The code of Figure~\ref{fig:code8} is a $(\sqrt{8},3-\sqrt{8})$-identifying code of density $\frac{4}{21}$. Here checking that the horizontal and diagonal patterns are intersecting the code is not enough. Indeed for $(u,v)$ with $v=u+(1,2)$, $\S_{\sqrt{8},3-\sqrt{8}}(u,v)$ is not containing any diagonal, horizontal or vertical pattern.



\appendix
\section*{Appendix A: Proof of Proposition~\ref{prop:method}}

In this appendix, we give an outline of the proof of Proposition~\ref{prop:method}:
%
\begin{proof}
We assume that the conditions of the statement of the proposition are satisfied.

First note that there are at most $2\lfloor r \rfloor+2$ different horizontal lines in $\S_{r,\Delta}((0,0),(-1,-1))$, so necessarily $t\leq 2r+2$.

The set $L^d_s$ corresponds to a code with diagonal lines repeated modulo $s$. The code $C=L^d_s \cup L^h_t$ has density $\frac{1}{s}+\frac{1}{t}-\frac{1}{st}$.
We will now prove that $C$ is an $(r,\Delta)$-identifying code.

We have $s\leq diag(B_r((0,0)))$, so $C$ is a dominating set.
Let $u$ and $v$ be two vertices, $u\neq v$. Without loss of generality, we can assume that $v=u+(x,y)$, with $x\geq 0$.
By Lemma~\ref{lem:line}, $C\cap \S_{r,\Delta}(u,v)$ is nonempty when $u$ and $v$ lie on the same horizontal or vertical line. Hence, we can assume that $x\geq 1$ and $|y|\geq 1$.

Let $k$ be the maximum positive ordinate of a vertex with abscissa $\lfloor r \rfloor$ in $B_r((0,0))$. Condition $(c)$ says that $k\geq 2$. Let $E=\{(x_1,y_1),\ldots,(x_t,y_t)\}$ be a set of vertices of  $\S_{r,\Delta}((0,0),(-1,-1))$ such that $y_i\equiv i \bmod t$.  We can assume that $x_iy_i\geq 0$ for all $i$. Indeed, assume for example that $x_i<0$ and $y_i>0$. If $(x_i,y_i)\in B_r((0,0))\setminus B_{r+\Delta}((-1,-1))$, then $(-x_i,y_i)\in B_r((0,0))\setminus B_{r+\Delta}((-1,-1))$ and we change $(x_i,y_i)$ to $(-x_i,y_i)$. Otherwise, $(x_i,y_i)\in B_r((-1,-1))\setminus B_{r+\Delta}((0,0))$, then $(-x_i+1,y_i)\in B_r((0,0))\setminus B_{r+\Delta}((-1,-1))$, and we change $(x_i,y_i)$ to $(-x_i+1,y_i)$.

This implies that for $0\leq i\leq k$,
$E+(0,-i)$ is included in  $\S_{r,\Delta}((0,0),(-1,-2i))$ and $\S_{r,\Delta}((0,0),(-1,-2i-1))$. Therefore, by translation,  $\S_{r,\Delta}(u,v)$ is intersecting all the horizontal lines modulo $t$ for $|y|\leq 2k+1$ and $x=1$, and by symmetry, this is also true for any $x>0$.
If $x>0$ and $y>2k+1$, then it is clear that $\S_{r,\Delta}(u,v)$ is intersecting all the diagonal lines modulo $diag(B_r(0,0))$; so $\S_{r,\Delta}(u,v)\cap C$ is nonempty.

We assume now that $x>0$ and $y<-(2k+1)$.
We first deal with the case $x=1$. Without loss of generality, we can assume that $u=(0,h)$ and $v=(1,-h')$, with $h'+1\geq h\geq h'\geq k+1$. If $h\geq r$, the diagonal line $y=x$ is not intersecting the circle $\C_{r+\Delta}(v)$ in a positive abscissa, and then it is clear that $\S_{r,\Delta}(u,v)$ is intersecting at least $diag(B_r((0,0)))$ diagonal lines. In the other case, as in the proof of Lemma~\ref{lem:line}, we will show that the distance on the diagonal line $y=x$ between the circle $\C_r(u)$ and the circle $\C_{r+\Delta}(v)$ is at least $\sqrt{2}$. That will imply that $\S_{r,\Delta}(u,v)$ is intersecting at least $diag(B_r((0,0)))$ consecutive diagonal lines. Let say that the diagonal line $y=x$ is intersecting $\C_r(u)$ in a vertex $(x_u,x_u)$ and $\C_{r+\Delta}(v)$ in a vertex $(x_v,x_v)$.
Then $x_u=\frac{h}{2}+\frac{\sqrt{2r^2-h^2}}{2}$ and $x_v= 1-\frac{h'}{2}+\frac{\sqrt{2r^2-h'^2+4r\Delta-2h'+2\Delta^2-1}}{2}$. Hence:
\[
x_v-x_u = \frac{h+h'}{2} -1 +\frac{\sqrt{2r^2-h^2}-\sqrt{2r^2-h'^2+4r\Delta-2h'+2\Delta^2-1}}{2} \geq h'-2 \geq 1 \textrm{.}
\]


Therefore, the distance between $(x_u,x_u)$ and $(x_v,x_v)$ is at least $\sqrt{2}$ and we are done.

With the same method, we can show that $\S_{r,\Delta}(u,v)$ is also intersecting $diag(B_r((0,0)))$ consecutive diagonal lines for $x\geq 2$.
\end{proof}

\section*{Appendix B: Proof of Proposition~\ref{prop:low2}}
In this appendix, we prove the proposition~\ref{prop:low2}:
\begin{proof}
We consider the pattern $F$ of Figure~\ref{fig:frame3}. We say that a pattern $u+F$ is a {\em neighbour} of a pattern $v+F$ if $u$ and $v$ are neighbours, i.e. $d(u,v)=1$. Then a pattern has four neighbours: on the left, on the right, on the bottom and on the top.
Let $\mathcal C$ be a $(\sqrt{2},2-\sqrt{2})$-identifying code. We will show that in average, there are at least $\frac{64}{19}$ vertices of $\mathcal C$ in each pattern, proving the result.

We first notice, that there are at least three vertices of the code in each pattern $u+F$ and there are only five possibilities, up to symmetry and rotations, that are shown on Figure~\ref{fig:possframe3}.
For $i\geq 3$, let $\mathcal F_i$ be the set of all patterns $u+F$,$u\in \Z^2$, such that $|\mathcal C\cap(u+F)|=i$. We define as well $\mathcal F_{\geq i}$ as the union of set $\mathcal F_{j}$ with $j\geq i$.
We call a pattern $A$ (resp. $B,C,D,E$) if it corresponds to the case of Figure~\ref{fig:F3a} (resp.~\ref{fig:F3b},\ref{fig:F3c},\ref{fig:F3d},\ref{fig:F3e}).
We call a pattern $A^+$ if it is a frame $A$ and if it has three neighbours in $\mathcal F_3$.
We will show that each frame of $\mathcal F_3$ has as neighbours:
\begin{itemize}
\item a frame in $\mathcal F_{\geq 5}$ and a frame in $\mathcal F_{\geq 4}$, or,
\item a frame in $\mathcal F_{\geq 5}$, and it is a frame $C$, or,
\item three frames in $\mathcal F_4$ or,
\item two frames in $\mathcal F_4$ with one of them having at most three neighbours in $\mathcal F_3$, or,
\item one frame in $\mathcal F_4$ and, at distance exactly $\sqrt{5}$, a frame in $\mathcal F_{\geq 5}$ and it is a $A+$.
\end{itemize}

In order to prove this, we consider the different types of frames and study their neighbourhoods. We always choose for the starting frame the orientation of Figure~\ref{fig:possframe3}. We first notice that except frames $A$ and $C$ all the frames have at least two neighbours in $\mathcal F_{\geq 4}$, on the left and on the top.

\begin{itemize}
\item[Frame $A$ :] A frame $A$ must have at least one neighbour in $\mathcal F_{\geq 4}$ (on the left). Assume first that there is only one (case $A+$), then necessarily the right neighbour of the frame is also a frame $A$.
  For domination reasons, there must be at least one code vertex in position $1$ or $2$ in Figure~\ref{fig:squareA+}.
\begin{itemize}
\item If there is a code vertex in position $2$ (and maybe one in position $1$), then the top-neighbour is necessarily a frame $E$ and there is a contradiction for the diagonal pattern in the top-right neighbour.
\item Hence, the code vertex is necessarily  in position $1$, then the bottom-neighbour is a frame $D$. The frame in bottom-right is in $\mathcal F_4$ and is fixed. The top-neighbour is necessarily a frame $D$ and the top-right neighbour is in $\mathcal F_4$ and is fixed. Then the right-right-up neighbour has at least five code vertices. We say that this frame of $\mathcal F_{\geq 5}$ and the basic frame $A+$ are {\em $A$-associated}. Note that the three first columns of the frame of $\mathcal F_{\geq 5}$ are completly fixed and give the direction of the frame $A+$, so a frame in $\mathcal F_{\geq 5}$ can be $A$-associated to at most one frame $A+$.
Moreover, the frame of $\mathcal F_{\geq 5}$ has at most two neighbours in $\mathcal F_3$, in top and right position. They cannot be frames $C$ or $A+$.
\end{itemize}

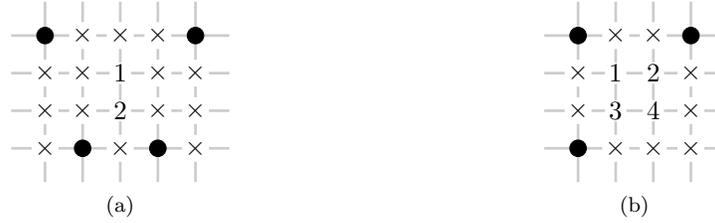
\begin{figure}[h]
\begin{center}
\subfloat[][\label{fig:squareA+}]{
\begin{tikzpicture}[scale=0.5]
\draw[mygrid] (-2.9,-2.9) grid (2.9,1.9);

\foreach \pos in {(-2,1),(-1,-2),(1,-2),(2,1)}
	{\node[code] at \pos {};}
	
\foreach \pos in {(-2,0),(-2,-1),(-2,-2),(-1,-1),(-1,0),(-1,1),(0,1),(0,-2),(1,-1),(1,0),(1,1),(2,0),(2,-1),(2,-2)}
	{\node[labelnode] at \pos {$\times$};}
	
	\node[labelnode] at (0,0) {$1$};
	\node[labelnode] at (0,-1) {$2$};

\end{tikzpicture}
}
\hfil
\subfloat[][\label{fig:squareB}]{
\begin{tikzpicture}[scale=0.5]
\draw[mygrid] (-2.9,-2.9) grid (1.9,1.9);

\foreach \pos in {(-2,1),(-2,-2),(1,1)}
	{\node[code] at \pos {};}
	
\foreach \pos in {(-2,0),(-2,-1),(-1,-2),(-1,1),(0,1),(0,-2),(1,-1),(1,0),(1,-2)}
	{\node[labelnode] at \pos {$\times$};}
	
	\node[labelnode] at (-1,0) {$1$};
	\node[labelnode] at (-1,-1) {$3$};
	\node[labelnode] at (0,0) {$2$};
	\node[labelnode] at (0,-1) {$4$};

\end{tikzpicture}
}
\end{center}
 \caption{Black dot is a vertex of the code and $\times$ means there is no code vertex.}
 \end{figure}

Assume now that the basic frame $A$ has only two neighbours in $\mathcal F_4$ with all of their neighbours in $\mathcal F_3$, and that the two other neighbours of the basic frame are in $\mathcal F_3$. One of the two neighbours in $\mathcal F_4$ is on the left. Assume first that the second one is on the right.
Then the frame in top of frame $A$ must be in $\mathcal F_3$ and three of its neighbours are frames of $\mathcal F_3$, so it must be a frame $A$ or $C$. It cannot be frame $A$, so it must be a frame $C$ and it is completely fixed. For the bottom neighbour, it must also be a frame $A$ or $C$, and it can only be a frame $C$. Then we have a contradiction because one vertex is not dominated.
Hence, we can now assume  that the  right neighbour of the basic frame is in $\mathcal F_3$. Then it is necessarily a frame $A$ and as before, there must be at least one code vertex among position $1$ and $2$.
\begin{itemize}
\item If there is a code vertex in position $2$, then the top-neighbour is in $\mathcal F_{\geq 4}$, and the top-right neighbour is also in $\mathcal F_{\geq 4}$, a contradiction.
\item Otherwise there must be a code vertex in position $1$. The top-left neighbour must be in $\mathcal F_3$. It cannot be a frame $E$ because a diagonal pattern would not be caught on position top-right. Also, the top-left neighbour is a frame $A$. Then the top-neighbour is also in $\mathcal F_3$ and is a frame $D$. The bottom-left neighbour is also in $\mathcal F_3$ and is fixing the bottom neighbour which is in $\mathcal F_4$. But then the bottom-right is in $\mathcal F_{\geq 4}$, a contradiction.
\end{itemize}

\item[Frame $B$:] We show that a frame $B$ must have at least three neighbours  in $\mathcal F_4$ or one in $\mathcal F_4$ and one in $\mathcal F_{\geq 5}$.
Assume that is not the case, then necessarily, the frame $B$ must have exactly two neighbours in $\mathcal F_4$ (the left and top-neighbours), the others must be in $\mathcal F_3$.
We denote the central squares of the frame with $1$,$2$,$3$,$4$ (see Figure~\ref{fig:squareB}). At least one of those vertices must be in the code for domination reasons.
\begin{itemize}
\item If there is a code vertex in position $1$, then the bottom and right frames are necessarily frame $D$, and then a diagonal pattern in position bottom-right is not caught.
\item If the code vertex is in position $2$, then the bottom frame is a frame $E$, but then the left frame must be in $\mathcal F_{\geq 5}$.
\item By symmetry the code vertex is not in position $3$
\item Finally, if there is a code vertex in position $4$ with no code vertices in positions $1$ to $3$, then the four code vertices in position left and top are completely determined, and with a contradiction for the diagonal pattern in position top-left.
\end{itemize}
\item[Frame $C$ :] A frame $C$ has a neighbour (the top one) in $\mathcal F_{\geq 5}$, and the top-top-one must be in $\mathcal F_{\geq 4}$. We say that the frame $C$ is {\em $C$-associated} to the top-neighbour in $\mathcal F_{\geq 5}$. Clearly, a frame in $\mathcal F_{\geq 5}$ is $C$-associated to at most two frames $C$.
\item[Frame $D$ :] Assume that a frame $D$ has only two neighbours in $\mathcal F_4$ (the top- and left-ones) and two neighbours in $\mathcal F_3$. We will show that one of the neighbours in $\mathcal F_4$ must have at most three neighbours in $\mathcal F_3$. We numerate as before the vertices of the central square $1$ to $4$. By contradiction, we assume that the neighbours in $\mathcal F_4$ have all their neighbours in $\mathcal F_3$.
\begin{itemize}
\item If there is a code vertex in position $1$, then on the right, there must be a frame $C$ and then the frame of $\mathcal F_4$ on the top has a neighbour frame in $\mathcal F_5$.
\item If there is a code vertex in position $2$ and none in position $1$, then the bottom frame is a frame $C$, the left frame, in $\mathcal F_4$, is fixed. For domination reasons, there must be a code vertex in position $3$, and the top-left must be in $\mathcal F_{\geq 4}$.
\item If there is a code vertex in position $3$ and none in positions $1$ and $2$, then it must be a frame $A$ on the right, then on top-right it must also be a frame $A$, then the top-frame is fixed, and the left-up must be in $\mathcal F_{\geq 4}$.
\item Otherwise, there must be a code vertex in position $4$, and none in positions $1$ to $3$, the top-frame is completely fixed, and its right neighbour is in $\mathcal F_{\geq 4}$.
\end{itemize}
\item[Frame $E$:] Assume that a frame $E$ has only two neighbours  in $\mathcal F_4$ (in top and left positions) and two  in $\mathcal F_3$. We will show that one of the neighbours  in $\mathcal F_4$ has at most three neighbours in $\mathcal F_3$. We numerate as before the vertices of the central square $1$ to $4$. Here we also have the possibilities that none of the four vertices is in the code. We assume that the two neighbours of $\mathcal F_4$ have only neighbours in $\mathcal F_3$.
\begin{itemize}
\item If there is one code vertex in position $1$,  there must a be frame $A$ on right and bottom, and frame of  $\mathcal F_4$ on top-right.
\item If there is one code vertex in position $2$, then the bottom one is a frame $C$, and then the left-frame is fixed, the top-top-frame must be in $\mathcal F_3$, and will necessarily be a frame $D$. Then the top-left is necessarily a frame $A$ and we have a contradiction for the top-frame.
\item By symmetry, it is the same for position $3$,
\item If there is only a code vertex in position $4$, then the left frame is fixed, and then the bottom-left must be in $\mathcal F_4$.
\item If there is no code vertices in the square, then the left-left-frame, which must be in $\mathcal F_3$, is a frame $A$ and fixes the left-frame. The same holds for the top-frame, but then there are at least four vertices in the top-left-frame.
\end{itemize}
\end{itemize}

We can show from the previous analysis that a frame in $\mathcal F_{\geq 5}$ has:
\begin{itemize}
\item one $A$-associated $A+$ and then at most two neighbouring frames in $\mathcal F_3$. They are not frame $C$ or $A+$ and so have as neighbour at least one other frame in $\mathcal F_{\geq 4}$, or,
\item no $A$-associated frame, two $C$-associated frames and no other neighbouring frames in $\mathcal F_3$, or,
\item no $A$-associated frame, one $C$-associated frame and at most two other neighbouring frames in $\mathcal F_3$, or,
\item no $A$-associated frame, no $C$-associated frame and at most four other neighbouring frames in $\mathcal F_3$.
\end{itemize}

Let now $\alpha = \frac{12}{19}$. We give charge $i$ to a frame of $\mathcal F_i$. We apply the following discharging rules :
\begin{enumerate}
\item A frame in $\mathcal F_4$ gives $\frac{\alpha}{4}$ to each neighbour in $\mathcal F_3$  different from frame $A+$, if there are four of them, and  $\frac{\alpha}{3}$ if there are at most three.
\item A frame in $\mathcal F_{\geq 5}$, gives $\frac{7}{12}\alpha$ to $A$- and $C$-associated frames and $\frac{\alpha}{3}$ to all its remaining neighbours in $\mathcal F_3$ at distance $1$.
\end{enumerate}

Then, after applying those rules, every frame has at least charge $4-\alpha=3+\frac{7}{12}\alpha=\frac{64}{19}$. Indeed:
\begin{itemize}
\item A frame in $\mathcal F_4$ gives at most $\alpha$ so ends with at least $4-\alpha$.
\item A frame  in $\mathcal F_{\geq 5}$ gives at most $\max\left(\frac{14}{12}\alpha,(\frac{7}{12}+\frac{2}{3})\alpha,
\frac{4}{3}\alpha\right) = \frac{4}{3}\alpha$, and so ends with at least $5-\frac{4}{3}\alpha>4-\alpha$.
\item A frame $A+$ receives $\frac{7}{12}\alpha$ and so ends with  $3+\frac{7}{12}\alpha$.
\item A frame $C$ receives at least $\frac{7}{12}\alpha$ from a frame of $\mathcal F_5$.
\item The other frames of $\mathcal F_3$ receives at least $\frac{\alpha}{4} +\frac{\alpha}{3}=\frac{7}{12}\alpha$.
\end{itemize}
\end{proof}

\end{document}